\documentclass[10pt,journal,a4paper]{IEEEtran}
\usepackage{amssymb,amsmath}
\usepackage{cite}
\usepackage{graphicx,subfigure}
\usepackage{psfrag}
\usepackage{url}
\usepackage[latin1]{inputenc}
\usepackage[absolute,overlay]{textpos}
\newcommand{\argmax}[1]{{\operatorname{arg}\,\max_{#1}}\,}

\usepackage{pgf}
\usepackage[latin1]{inputenc}

\usepackage{algorithm}
\usepackage{algpseudocode}

\usepackage{amsthm}

\newtheorem{lemma}{Lemma}
\newtheorem{theorem}{Theorem}
\newtheorem{corollary}{Corollary}
\newtheorem{remark}{Remark}

\begin{document}

\title{Joint Power Control and Rate Allocation enabling Ultra-Reliability and Energy Efficiency in SIMO Wireless Networks}

\author{
	\IEEEauthorblockN{Onel L. Alcaraz L\'{o}pez, \IEEEmembership{Student Member, IEEE},
	Hirley Alves, \IEEEmembership{Member, IEEE}, \\ 
	Matti Latva-aho, \IEEEmembership{Senior Member, IEEE}
}
	\thanks{Authors are with the Centre for Wireless Communications (CWC), University of Oulu, Finland. \{onel.alcarazlopez, hirley.alves, matti.latva-aho\}@oulu.fi.}	
	\thanks{This research has been financially supported by Academy of Finland, 6Genesis Flagship (Grant n.318937) and ee-IoT (Grant n.319008), and Academy Professor (Grant n.307492), and the Finnish Funding Agency for Technology and Innovation (Tekes), Bittium Wireless, Keysight Technologies Finland, Kyynel, MediaTek Wireless, Nokia Solutions and Networks.}
}

\maketitle

\begin{abstract}
	Coming cellular systems are envisioned to open up to new   services with stringent reliability and energy efficiency requirements. In this paper we focus on the joint power control and rate allocation problem in Single-Input Multiple-Output (SIMO) wireless systems with Rayleigh fading and stringent reliability constraints. We propose an allocation scheme that maximizes the energy efficiency of the system while making use only of average statistics of the signal and interference, and the number of antennas $M$ that are available at the receiver side.  
	We show the superiority of the Maximum Ratio Combining (MRC) scheme over Selection Combining (SC) in terms of energy efficiency, and prove that the gap between the optimum allocated resources converges to $(M!)^{1/(2M)}$ as the reliability constraint becomes more stringent. Meanwhile, in most cases MRC was also shown to be more energy efficient than Switch and Stay Combining (SSC) scheme, although this does not hold only when operating with extremely large $M$, extremely high/small average signal/interference power and/or highly power consuming receiving circuitry.
	 Numerical  results show the feasibility of the ultra-reliable operation when $M$ increases, while greater the fixed power consumption and/or drain efficiency of the transmit amplifier is, the greater the optimum transmit power and rate.
\end{abstract}
\section{Introduction}\label{int}
The advent of fifth generation (5G) of wireless systems opens up new possibilities and gives rise to new use cases with stringent reliability requirements, e.g.,  Ultra Reliable Low Latency Communication paradigm (URLLC) \cite{Popovski.2014}. Some examples are \cite{Schulz.2017}:
 factory automation, where the maximum error probability should be around
$10^{-9}$; smart grids ($10^{-6}$), professional audio ($10^{-6}$), etc.
Meeting such requirements is not an easy task and usually  various  diversity sources are necessary in order to attain the ultra-reliability region \cite{Popovski.2017}. The problem becomes even more complicated if stringent delay constraints have to be satisfied\footnote{In general, there is a fundamental trade-off between delay and reliability metrics due to the fact that by relaxing one of them, we can enhance the performance of the other. In fact, Long-Term Evolution (LTE) already offers guaranteed bit rate that can support packet error rates down to $10^{-6}$, however, the delay budget goes up to $300$ms including radio, transport and core network latencies \cite{Singh.2017}, which is impractical for many real time applications.}, and/or if power consumption is somewhat limited as is the case in systems of low-power devices such as sensors or tiny actuators. The interplay between the diverse requirements makes physical layer design of such systems very complicated \cite{Hyoungju.2017}.

The principles for supporting URLLC are discussed in \cite{Popovski.2017} by considering, for instance, the design of packets and access protocols. In \cite{Simsek.2016,Aijaz.2017} authors outline the key technical requirements and architectural approaches pertaining to wireless access protocols, radio resource management aspects, next generation core networking capabilities, edge-cloud, and edge artificial intelligence capabilities, and propose first avenues for
specific solutions to enable the Tactile Internet revolution. 
The trade-off between reliability, throughput, and latency, when transmitting short packets in a multi-antenna setup, is identified in  \cite{Durisi.2016}. Moreover, authors present some bounds that allow to determine the optimal number of transmit antennas and the optimal number of time-frequency diversity branches that maximize the rate. 
Shared diversity resources are explored deeply in \cite{Kotaba.2018}  when multiple  connections are only intermittently active, while cooperative communications are also considered in literature, e.g., \cite{Nouri.2017}, and \cite{Lopez.2017,Lopez2.2017} for wireless powered communications, as a viable alternative to direct communication setups  \cite{Lopez3.2017}.

Intelligent resource allocation strategies are of paramount importance to provide efficient ultra-reliable communications. In \cite{She.2017}, the network availability for supporting the quality of service of users is investigated, while some tools for resource optimization addressing the  delay and packet loss components in URLLC are presented. 
Energy-efficient design of fog computing networks supporting  Tactile Internet applications is the focus of the research in \cite{Xiao.2018} where the workload is allocated such that it minimizes the response time under the given power efficiency constraints of fog nodes; while in \cite{Hou.2018} authors propose a resource management protocol to meet the stringent delay and reliability requirements while minimizing the bandwidth usage. 
A power control protocol is presented in \cite{Lopez3.2018} for a single-hop ultra-reliable wireless powered system and the results show the enormous impact on improving the system performance, in terms of error probability and energy consumption. The minimum energy required to transmit $k$ information bits over a Rayleigh block-fading channel in a multi-antenna setup with no interference and with a given reliability is investigated in \cite{Yang.2017}.
On the other hand, link adaptation
optimization through an adaptive modulation and coding scheme, considering errors in both data and feedback channels, is proposed in \cite{Shariatmadari.2016}, and authors show that the performance of their proposed scheme approximates to the optimal.
An energy efficient power allocation strategy for the Chase Combining (CC) Hybrid Automatic Repeat Request (HARQ) and Incremental Redundancy (IR) HARQ  setup is suggested in \cite{Dosti.2017} and \cite{Avranas.2018}, respectively; while allowing to reach any outage probability target in the finite block-length regime. 
In \cite{Shehab.2017}, a detailed analysis of the effective energy efficiency for delay constrained networks in the finite blocklength regime is presented, and the optimum power allocation strategy is found. Results reveal that Shannon's model underestimates the optimum power when compared to the exact finite blocklength model. Authors in \cite{Farayev.2015} formulate a joint power control and discrete rate adaptation problem with the objective of minimizing the time required for the concurrent transmission of a set of sensor nodes while satisfying their delay, reliability and energy consumption requirements. In \cite{Lopez.2018} we focused on the rate allocation problem in downlink cellular networks with Rayleigh fading and stringent reliability constraints. The allocated rate depends on the target reliability, and on average statistics of the signal and interference and the number of antennas that are available at the receiver side. We have shown the feasibility of the ultra-reliable operation when the number of antennas increases, and also that the results remain valid even when operating with stringent delay constraints as far as the amount of information to be transmitted is not too small. The rate allocation strategy is extended to downlink Non-orthogonal multiple access (NOMA) scenarios in \cite{Lopez2.2018}, while we attain the necessary conditions so that NOMA overcomes the orthogonal multiple access (OMA) alternative. Additionally, we discuss the optimum strategies for the 2-user NOMA setup when operating with equal rate or maximum sum-rate goals. 

In this paper we develop further \cite{Lopez.2018} by generalizing some of its main results to the case where the transmit power is another degree of freedom that is exploited to meet the reliability requirements while maximizing the energy efficiency of the system. Therefore, we focus on joint power control and rate allocation strategies that maximize the system energy efficiency in ultra-reliable system with multiple antennas at receiver side, thus, a Single-Input Multiple-Output (SIMO) system.
There is no distinction between uplink and downlink, but notice that SIMO setups match much better uplink scenarios where the receiver is usually equipped with better hardware capabilities, e.g., data aggregators/gateways or base stations in cellular communications\footnote{Notice that some URLLC applications, e.g., tactile Internet, may require the joint design of downlink and uplink communications (check for instance \cite{She.2018}). Such analysis is out of the scope of this paper; however, as future work we intend to extend our results for the Multiple-Input Multiple-Output (MIMO) scenario, while considering the mentioned joint downlink and uplink design.}. The system is composed of an ultra-reliable link under Rayleigh fading, being interfered by multiple transmitters operating in the neighborhood, thus, differently from the setups analyzed in \cite{She.2017,Lopez3.2018,Yang.2017,Shariatmadari.2016,Dosti.2017,Farayev.2015,Lopez2.2018}. The main contributions of this work can be listed as follows:
 \begin{itemize}
 	\item we propose a joint power control and rate allocation scheme that meets the stringent reliability constraints of the system while maximizing the energy efficiency. The allocated resources depend only on the target reliability, and on average statistics of the signal and interference and the number of antennas that are available at the receiver side. In addition to Selection Combining (SC) and Maximum Ratio Combining (MRC) schemes, and different from \cite{Lopez.2018}, we also consider the Switch and Stay Combining (SSC) technique; while we do not make distinction between uplink and downlink and our goal is to maximize the energy efficiency of the system by adjusting both the transmit power and rate;	
 	\item  we attain accurate closed-form approximations for the resources, optimum transmit power and rate, to be allocated when the receiver operates using the SC, SSC and MRC schemes;
 	\item we show that the optimum transmit rate and power are smaller when operating with SSC than with SC, and the ratio gap tends to be inversely proportional to the square root of a linear function of the number of antennas $M$ at the receiver; however, such allocation provides always positive gains in the energy efficiency performance;
 	\item we show the superiority of MRC over SC in terms of energy efficiency, since it allows operating with greater/smaller transmit rate/power. We proved that the performance gap between the optimum allocated resources for these schemes in the asymptotic ultra-reliable regime, where the outage probability tends to 0, converges to $(M!)^{1/(2M)}$. Meanwhile, in most cases MRC was also shown to be more energy efficient than SSC, although this does not hold only when operating with extremely large $M$, extremely high/small average signal/interference power and/or highly power consuming receiving circuitry;
 	\item we show that the greater the fixed power consumption and/or drain efficiency of the transmit amplifier, the greater the optimum transmit power and rate. However, the energy efficiency decreases/increases with the power consumption/drain efficiency. Numerical  results also show the feasibility of the ultra-reliable operation when the number of antennas increases.
 \end{itemize}

Next, Section \ref{system} overviews the system model and assumptions. Section \ref{sec2} introduces the performance metrics and the optimization problem, while in Section \ref{sec3} we characterize the Signal-to-Interference Ratio (SIR) distribution for each of the receive combining schemes. In Section~\ref{joint} we find the resource allocation strategy that maximizes the system energy efficiency subject to stringent reliability constraints. Finally, Section~\ref{results} presents the numerical results and  Section \ref{conclusions} concludes the paper.

\textit{Notation:} Boldface lowercase letters denote vectors,  for instance, $\mathbf{x}=\{x_i\}$, where $x_i$ is the $i$-th element of $\mathbf{x}$.
$X\!\sim\!\mathrm{Exp}(1)$ is a normalized exponential distributed random variable with Cumulative Distribution Function (CDF)  $F_X(x)=1-e^{-x}$, $x\ge 0$, while $Y\!\sim\! L(p,q)$ is a Lomax random variable with Probability Density Function (PDF) $f_Y(y|p,q)\!=\!q\big(1\!+\!\frac{q}{p}y\big)^{-1\!-\!p}\!$ and  CDF $F_Y(y|p,q)\!=\!1\!-\!\big(1\!+\!\frac{q}{p}y\big)^{-p}$, $y\ge 0$, and $Z\sim \mathcal{P}(z|p,q)$ is a Pareto I random variable with PDF $f_Z(z|p,q)=p q^p z^{-p-1},\ z\ge q$. $\mathbb{P}[A]$ is the probability of event A, $\mathbb{E}[\!\ \cdot\ \!]$ denotes expectation, while $\lfloor x \rfloor$ denotes the largest integer that does not exceed $x$.
Also, $Q^{-1}(\cdot)$ is the inverse Q-function, 
$\Gamma(p,x)=\int_x^{\infty}t^{p-1}e^{-t}\mathrm{d}t$ is the incomplete gamma function, while $\mathcal{W}(x)$ is the main branch of the Lambert W function \cite{Corless.1996}, which satisfies $\mathcal{W}(x)\ge -1$ for $x\in \mathbb{R}$ and it is defined in $-1/e\le x<0$.
\section{System Model}\label{system}
\begin{figure}[t!]
	\centering
	\subfigure{\includegraphics[trim={2.4cm 0 0 0},clip,width=0.45\textwidth]{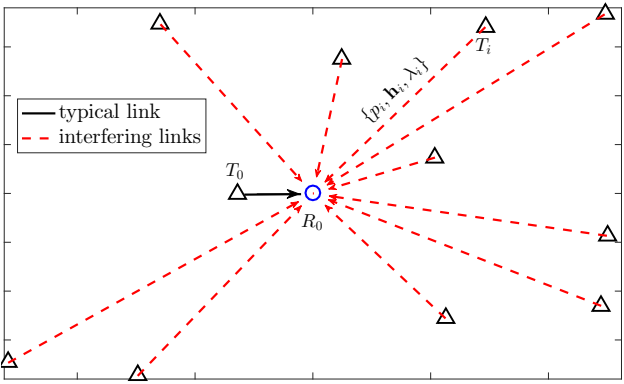}}
	\vspace*{-2mm}
	\caption{Illustration of the system model with $\kappa=10$ interfering nodes, where $p_i,\ \mathbf{h}_i$ and $\lambda_i$ are the transmit power of $T_i$, the power channel gain vector and path-loss from $T_i$ to $R_0$, respectively.}		
	\vspace*{-2mm}
	\label{Fig1}
\end{figure}
Consider the scenario in Fig.~\ref{Fig1}, where a collection of $\kappa+1$ nodes,  $T_i,\ i=0,1,\cdots,\kappa$, are spatially distributed  in a given area and using the same spectrum resources, e.g., time and frequency, when transmitting to their corresponding receivers. We focus on the performance of link $0$, which we refer to as the typical link, and denote $T_0$ and $R_0$ as its transmitter and receiver node, respectively; while the transmit rate is denoted as $r_0$. 
Meanwhile $T_i\rightarrow R_0$ with $i=1,\cdots,\kappa$ denotes each of the interfering links. 
We assume a SIMO setup where $R_0$ is equipped with $M$ antennas sufficiently separated such that the fading affecting the received signal in each antenna can be assumed independent and  Channel State Information (CSI) is available at $R_0$,\footnote{$T_0$ may send some pilot symbols as overhead when transmitting to $R_0$ for the latter be able to estimate the CSI. be able to estimate the CSI. Notice that this overhead can be accounted as part of the constraint $r_{\min}$ (check Section~\ref{sec2}); while although we assume perfect CSI, the imperfectness may be modeled as a loss in the SIR as in \cite{She.2018,Lopez2.2017}.} hence full gain from spatial diversity can be attained\footnote{Diversity is an important building block for supporting URLLC \cite{Popovski.2017}, and herein we focus simply on spatial diversity taking advantage of the multiple receive antennas. Notice that other diversity sources such as frequency, time and/or polarization could also be available \cite{Goldsmith.2005}, and our results and methodology can be easily re-utilized/extended to cover such scenario.}. Particularly, one of the following combining schemes is utilized at $R_0$:
	\begin{itemize}
		\item \textit{SC:} The combiner outputs the signal on the antenna branch with the highest SIR.  Since only one branch is used at a time, SC could require just one receiver that is switched into the active antenna branch. However, a dedicated receiver on each branch may be needed for systems that transmit continuously in order to simultaneously and continuously monitor SIR on each branch. In this work we refer specifically to the latter SC implementation.
		Notice that with SC the resulting SIR is equal to the maximum SIR of all branches \cite{Goldsmith.2005}; 
		\item \textit{SSC:} This scheme strictly avoids the need for a dedicated receiver on each branch, thus reducing the power consumption, by scanning each of the branches in sequential order and outputting  the first signal with SIR above a threshold. Once a branch is chosen, as long as the SIR on that branch remains above the desired threshold, the combiner outputs that signal; while when the SIR on the selected branch falls below the threshold, the combiner switches to another branch \cite{Goldsmith.2005};
		\item \textit{MRC:} The combiner outputs a weighted sum of the signals coming from all branches. We assume that $R_0$ can perfectly estimate also the interference power level in every branch, thus, the optimum combining weight for each branch using such information is obtained by correcting 	the phase-mismatch of the received signal and scaling it by the interference level. In this case the resulting SIR is equal to the sum of SIRs on each branch \cite{Tanbourgi.2014}.
	\end{itemize}
We focus our attention to above combining schemes, while other possibilities include the Equal Gain Combining (EGC), which co-phases the signals on each branch and then combines them with equal weighting; and several hybrid schemes \cite{Simon.2005}. In general, these schemes are easier to implement compared to MRC but also perform slightly worse in terms of reliability.\footnote{For instance, the error performance of EGC typically exhibits less than 1 dB of power penalty compared to MRC \cite{Goldsmith.2005}.} In any case, such schemes lead to cumbersome analytical analysis, which we leave for future work.

Additionally, each link is characterized by a triplet $\{p_i,\mathbf{h}_i,\lambda_i\},\ i=0,1,\cdots,\kappa$, where $p_i\in[p_{\min},p_{\max}]$ is the transmit power of $T_i$ which is constrained to be not smaller and not greater than $p_{\min}$ and $p_{\max}$, respectively; $\mathbf{h}_i=[h_{i,1}, h_{i,2},\cdots, h_{i,M}]$ is the power channel gain vector with normalized and exponentially distributed entries such that $h_{i,j}\sim\mathrm{Exp}(1)$, e.g., Rayleigh fading; while $\lambda_i$ is the path-loss of the link. Meanwhile, we consider an interference-limited wireless system given a dense spatial deployment where the impact of noise is neglected\footnote{However, the impact of the noise could easily be incorporated without substantial changes.}; 
thus, the SIR perceived in the $j-$th antenna of $R_0$ is 
\begin{align}
\mathrm{SIR}_j&=\frac{p_0\lambda_0h_{0,j}}{\sum_{i=1}^\kappa p_i \lambda_i h_{i,j}}. \label{snri}
\end{align}
\section{System Performance Targets}\label{sec2}
Our goal in this work is to allocate power and rate at $T_0$ in order to maximize the system energy efficiency while meeting the URLLC requirements. Therefore, let us define these performance metrics.
\subsection{Reliability \& Latency}\label{A}
Reliability is defined as the probability that a data of given size is successfully transferred within a given time period \cite{Bennis.2018}. Hence, reliability and latency are intrinsically connected concepts. In fact, the typical URLLC use case demands transmitting a layer 2 protocol data unit of 32 bytes within 1 ms with $1-10^{-5}$ success probability \cite{3GPP.2017}.

During the last years, significant progress has been made within the information theory community to address the problem of quantifying the achievable rate while accounting for stringent reliability and latency constraints in a satisfactory way. In that sense, works in \cite{Polyanskiy.2010,Yang.2013} have identified these trade-offs for both Additive White Gaussian Noise (AWGN) and fading channels, respectively. Specifically, authors in \cite{Polyanskiy.2010} show that to sustain a desired error probability $\varepsilon$ at a finite blocklength $n$, one pays a penalty on the rate (compared to the Shannon's channel capacity) that is proportional to $Q^{-1}(\varepsilon)/\sqrt{n}$; while under quasi-static fading impairments authors in \cite{Yang.2013} show through numerical evaluation that the convergence to the outage capacity is much faster as $n$ increases than in the AWGN case. In fact, it has been shown in \cite{Mary.2016} for Nakagami-m and Rice channels that quasi-static fading makes disappear the effect of the finite blocklength when i) the rate is not extremely small
and ii) line of sight parameter is not extremely large.	For the scenario under discussion in the current work we have already corroborated in \cite{Lopez.2018} that by using the asymptotic outage probability instead of the finite blocklength error probability as the reliability performance metric, the results remain valid as far as the transmission rate is not too small. Therefore, in this work we leave aside the finite blocklength formulation (although the same methodology as in \cite{Lopez.2018} can be utilized) and just consider the outage probability. Notice that by limiting $r_0$ to be above some $r_{\min}$, the latency constraint is implicitly considered.

Considering the receive diversity schemes discussed in previous section, an outage event as a function of $r_0$ and $p_0$ is defined as $\mathcal{O}(r_0,p_0)\stackrel{\triangle}{=}\{f(\mathbf{SIR})<2^{r_0}-1\}$, where
\begin{align}\label{fSIR}
f(\mathbf{SIR})=\left\{ \begin{array}{ll}
\max\limits_{j=1,\cdots,M} \mathrm{SIR}_j, &   \mathrm{for\ SC\ and\ SSC} \\
\sum\limits_{j=1}^{M}\mathrm{SIR}_j, & \mathrm{for\ MRC}
\end{array}
\right..
\end{align}
Notice that in delay-limited systems with fixed transmit rate $r_0$ as in our case both SC, and SSC with threshold $2^{r_0}-1$, share the same outage performance. This is because iff the maximum SIR exceeds the threshold $2^{r_0}-1$, SSC will find at least one antenna branch with SIR above it, hence, no outage.
Finally, the outage probability can not exceed a given reliability constraint specified by the maximum outage probability $\varepsilon\ll 1$. This is $\mathbb{P}[\mathcal{O}(r_0,p_0)]\le\varepsilon$.
\subsection{Energy Efficiency}\label{B}
The energy efficiency is defined as the ratio between the throughput and the power consumption and it tells us the number of bits that can be transmitted per Hertz while consuming a joule unit. Considering a linear power consumption model as in \cite{Shuguang.2004,Shehab.2017}, we can write the energy efficiency of the system as  
\begin{align}\label{EE}
\mathrm{EE}(r_0,p_0,\beta\!=\!\frac{r_0 \big(1-\mathbb{P}[\mathcal{O}(r_0,p_0)]\big)}{p_0/\eta\! +\! p_t\!+\!M^{\beta}p_r\!+\!p_{syn}}\ (bits/J/\mathrm{Hz}),
\end{align}
where $\eta$ is the drain efficiency of the amplifier at $T_0$, $p_{syn}$ is the power consumption value for the frequency synthesizers at $T_0$ and $R_0$,\footnote{For the case of $R_0$ we assume that the frequency synthesizer is shared among all the antenna paths, thus, the consumption of this block does not depend on $M$ \cite{Shuguang.2004}.} while $p_t$ and $p_r$ are  the power consumed by the remaining internal circuitry for transmitting and receiving, respectively. Additionally, 
\begin{align}\label{beta}
\beta=\left\{ \begin{array}{ll}
0, &   \mathrm{for\ SSC} \\
1, & \mathrm{for\ SC\ and\ MRC}
\end{array}
\right.
\end{align}
since for SC and MRC the consumption of the internal circuitry grows linearly with $M$ because all the antenna branches are active, while for SSC only one is active\footnote{For SSC we do not take into account the sleep-mode power consumption of the circuitry in the inactive antenna branches, neither the power consumption when scanning the antennas trying to find one that provides a SIR value above the threshold $2^{r_0}-1$. Hence, the real power consumption of SSC may exhibit a weak dependence on $M$ but we ignore it here for simplicity, then, the energy efficiency performance of the SSC discussed here can be seen as an upper bound for the performance of a practical SSC implementation.}. 
\subsection{Problem Formulation}\label{C}
According to the performance metrics specified in Subsection~\ref{A} and \ref{B} we present in \eqref{eq4} the joint power control and rate allocation problem that maximizes the energy efficiency subject to an ultra-reliability constraint.
\begin{subequations}\label{eq4}
	\begin{alignat}{2}
	\mathbf{P1:}\qquad &\argmax{p_0,r_0}       &\qquad& 
	\mathrm{EE}(r_0,p_0,\beta) \label{eq4:a}\\
	&\text{s.t.} &      & \mathbb{P}[\mathcal{O}(r_0,p_0)]\le \varepsilon,\label{eq4:b}\\
	&                  &      & r_0\ge r_{\min},\label{eq4:c}\\
	& & & p_{\min}\le p_0 \le p_{\max}. \label{eq4:d}
	\end{alignat}
\end{subequations}
We would like to point out that the constraints on $p_0$ may be given by hardware limitations but also/alternatively $p_{\max}$ could be chosen to guarantee that certain interference thresholds on neighboring networks are not overpassed. Additionally, and as commented before in Subsection~\ref{A}, a delay constraint $t_{\max}$ can be implicitly considered within $r_{\min}$ by setting $r_{\min}=D/(B\times t_{\max})$ where $B$ (Hz) is the  bandwidth and $D$ (bits) is the data to be transmitted.

Fig.~\ref{Fig2} shows the feasible region when solving \textbf{P1}. As $p_0$ increases $T_0$ is capable of transmitting with a larger bit rate for the same reliability target, thus, the curve $r_0$ vs $p_0$ with $\mathbb{P}[\mathcal{O}(r_0,p_0)]=\varepsilon$ is increasing on $p_0$ as shown in the figure. Let us focus the attention on the red point on the curve $\mathbb{P}[\mathcal{O}(r_0,p_0)]=\varepsilon$, and notice that for any positive $\alpha_1$ and $\alpha_2$, $\mathbb{P}[\mathcal{O}(r_0-\alpha_1,p_0+\alpha_2)]<\varepsilon$ holds, but according to \eqref{EE} and based on the fact that $\varepsilon\ll 1$ we have that  $\mathrm{EE}(r_0-\alpha_1,p_0+\alpha_2,\beta)<\mathrm{EE}(r_0,p_0,\beta)$, thus, the solution of \textbf{P1} lies on the curve $\mathbb{P}[\mathcal{O}(r_0,p_0)]=\varepsilon$. Additionally, \textbf{P1} has a non-empty solution when $\mathbb{P}[\mathcal{O}(r_{\min},p_{\max})]\le\varepsilon$.
\begin{figure}[t!]
	\centering
	\subfigure{\includegraphics[trim={0.2cm 0 0 0},clip,width=0.45\textwidth]{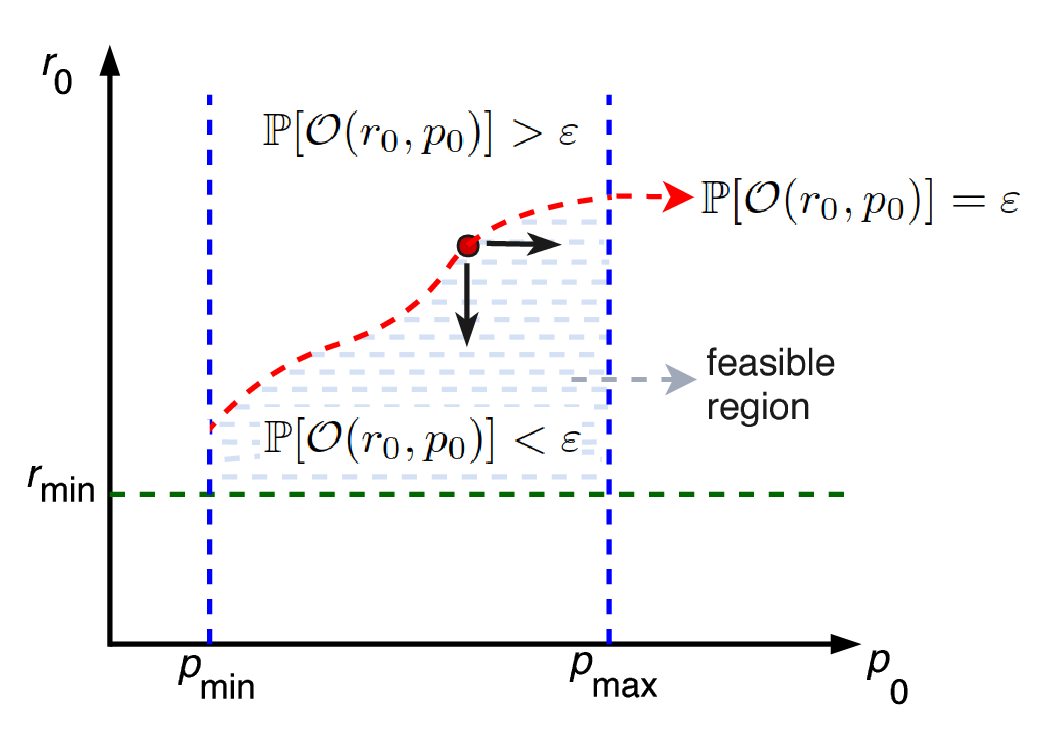}}
	\vspace*{-2mm}
	\caption{Feasible region for $\mathbf{P1}$.}			
	\label{Fig2}
	\vspace*{-2mm}
\end{figure}
Notice that the solutions of \textbf{P1}, named $p_0^*$ and $r_0^*$, must depend on information easy to obtain for $T_0$. For instance, it is not practical if $r_0$ and/or $p_0$ are set according to the interference contribution of each interfering node separately.
\section{SIR Distribution}\label{sec3}
Instantaneous channel fluctuations are unknown at $T_0$, thus, $r_0^*$ and $p_0^*$ are chosen fixed. Notice that in order to solve \textbf{P1} we first need to characterize the SIR distribution under each of the diversity schemes since
\begin{align}
\mathbb{P}[\mathcal{O}(r_0,p_0)]&=F_{f(\mathbf{SIR})}(2^{r_0}-1)=\varepsilon,\label{F}\\
r_0&=\log_2\big(1+F_{f(\mathbf{SIR})}^{-1}(\varepsilon)\big).\label{Finv}
\end{align} 
We proceed by finding the distribution of the $\mathrm{SIR}$ at each antenna and then we extend the results for multiple antennas at the receiver and under the SC, SSC and MRC schemes.
\begin{theorem}\label{the1}
	The CDF of the SIR at each antenna $j=1,...,M$ is given by
	\begin{align}
	F_{\mathrm{SIR}_j}(\gamma)=1-\prod_{i=1}^\kappa\frac{1}{1+ \frac{p_i\lambda_i}{p_0\lambda_0}\gamma},\label{Fsnr}
	\end{align}	
	which is upper-bounded by
	\begin{align}
	F_{\mathrm{SIR}_j}(\gamma)\le 1-\Big(1+\frac{\gamma}{\kappa\delta p_0}\Big)^{-\kappa}  \label{Fsnrapp} 
	\end{align}
	with  $\delta=\frac{\lambda_0}{\sum_{i=1}^{\kappa}p_i\lambda_i}$. 
\end{theorem}
\begin{proof}
	We proceed as follows \cite{Lopez.2018}
	\begin{align}
	F_{\mathrm{SIR}_j}(\gamma)&=\mathbb{P}\big[\mathrm{SIR}_j<\gamma\big]=1-\mathbb{P}\big[\mathrm{SIR}_j>\gamma\big]\nonumber\\
	&\stackrel{(a)}{=}1-\mathbb{P}\Big[\frac{p_0 \lambda_0 h_{0,j}}{\sum_{i=1}^{\kappa}p_i\lambda_ih_{i.j}}>\gamma \Big]\nonumber\\
	&=1-\mathbb{P}\Big[h_{0,j}>\frac{\gamma \sum_{i=1}^{\kappa}p_i\lambda_ih_{i.j}}{p_0\lambda_0}\Big]\nonumber\\
	&\stackrel{(b)}{=}1-\mathbb{E} \Big[ e^{-\frac{\gamma \sum_{i=1}^{\kappa}p_i\lambda_ih_{i.j}}{p_0\lambda_0}} \Big]\nonumber\\
	&=1-\prod_{i=1}^\kappa\mathbb{E} \big[e^{-\frac{\gamma p_i\lambda_ih_{i.j}}{p_0\lambda_0}}\big], \label{FsnrP}
	\end{align}
	where $(a)$ follows from using \eqref{snri}, $(b)$ comes from the complementary CDF of exponential random variable $h_{0,j}$, and \eqref{Fsnr} comes directly after \eqref{FsnrP}. Now we focus on the upper bound.
	\begin{align}
	\prod_{i=1}^\kappa\Big(1+\frac{p_i\lambda_i}{p_0\lambda_0}\gamma\Big)&\stackrel{(a)}{\le}\left[\frac{1}{\kappa}\sum\limits_{i=1}^{\kappa}\big(1+\frac{p_i\lambda_i}{p_0\lambda_0}\gamma\big)\right]^{\kappa}\nonumber\\
	&\stackrel{(b)}{=}\bigg[1+\frac{\gamma \sum_{j=1}^{\kappa}p_i\lambda_i}{\kappa p_0\lambda_0}\bigg]^{\kappa}\nonumber\\
	&\stackrel{(c)}{=}\bigg[1+\frac{\gamma}{\kappa\delta p_0}\bigg]^{\kappa},\label{ap2}
	\end{align}
	where $(a)$ comes from  using the  relation between the geometric and the arithmetic mean, $(b)$ follows from simple algebraic transformations, and $(c)$ by adopting $\delta=\frac{\lambda_0}{\sum_{i=1}^{\kappa}p_i\lambda_i}$. Substituting \eqref{ap2} into \eqref{Fsnr} we attain \eqref{Fsnrapp}.	
\end{proof}
\begin{remark}\label{Re1}
	Both, \eqref{Fsnr} and \eqref{Fsnrapp}, converge in the left tail. This becomes evident from the proof of Theorem~\ref{the1}. Therein notice that when operating in the left tail $\prod_{i=1}^{\kappa}\big(1+\frac{p_i\lambda_i}{p_0\lambda_0}\gamma \big)$ should be close to 1, therefore each of the terms $\big(1+\frac{p_i\lambda_i}{p_0\lambda_0}\gamma \big)\ge 1$ is expected to approximate to the unity. Hence, all of these terms are very similar among one another, and geometric mean approximates heavily  to arithmetic mean  in such scenarios. 
	\begin{figure}[t!]
		\centering
		\subfigure{\includegraphics[width=0.47\textwidth]{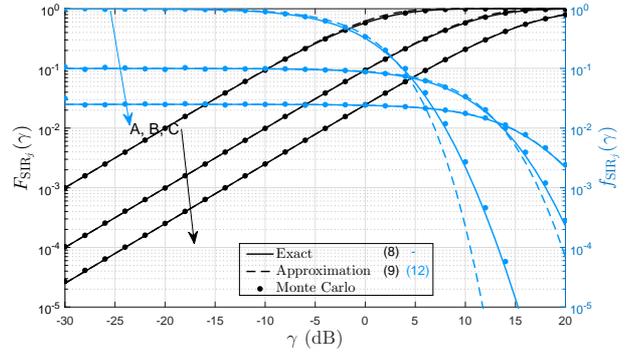}}
		\vspace*{-2mm}
		\caption{Comparison between the exact and approximate expressions of $F_{\mathrm{SIR}_j}(\gamma)$ and $f_{\mathrm{SIR}_j}(\gamma)$ for $p_i\lambda_i=2^{-i}\ \mu$W, $i=1,\cdots,\kappa$ and three different setups: $A:\ p_0\lambda_0=1\ \mu$W, $\kappa=18$; $B:\ p_0\lambda_0=10\ \mu$W, $\kappa=8$; $C:\ p_0\lambda_0=30\ \mu$W, $\kappa=2$. Monte Carlo simulation with $10^7$ samples is also shown. The exact PDF was evaluated taking the derivative of the exact CDF \eqref{Fsnr} for each setup.}	
		\vspace*{-2mm}	
		\label{Fig3}
	\end{figure}	
\end{remark}
\begin{corollary}\label{cor1}
	As a consequence of \eqref{Fsnrapp} and Remark~\ref{Re1}, the SIR at each antenna $i=1,...,M$ is approximately a Lomax random variable with PDF given by
	\begin{align}
	f_{\mathrm{SIR}_i}(\gamma)&\approx \frac{1}{\delta p_0}\Big(1+\frac{\gamma}{\kappa\delta p_0}\Big)^{-\kappa-1}.\label{pdf0}
	\end{align}		
	This can be represented as a scaled Lomax distribution such that $\mathrm{SIR}_j\!\approx\! \kappa\delta p_0 \varphi_j$ with $\varphi_j\!\sim\! L(\kappa,\!1)$.
\end{corollary}
	The convergence of the approximations in the left tail is clearly illustrated in Fig.~\ref{Fig3} for three different setups, thus, validating our findings. Additionally, notice that the exact CDF of $\mathrm{SIR}_j$ is upper-bounded by the approximation in the entire region, but this does not hold for the PDF in the right tail, for which the approximation lies under the exact curve and diverging fast.
\begin{remark}\label{Re0}
	Obtaining the PDF of the SIR directly from \eqref{Fsnr} seems intractable for large $\kappa$, which is the case in dense network deployments. Also, since the upper bound is extremely tight in the left tail of the distribution, its utility is enormous because it is in that region where typical reliability constraints are, e.g., $\varepsilon<10^{-1}$.
\end{remark}
\begin{remark}\label{Re2}
	Notice that \eqref{Fsnrapp} and \eqref{pdf0} depend only on the number of interfering nodes and $\delta p_0$, which is the ratio between the average signal and the average interference powers. These parameters can be easily obtained, specially for static or quasi-static deployments.
\end{remark}
\subsection{Selection Combining -- SC \& Switch and Stay Combining -- SSC}
Under the SC and SSC schemes, \eqref{F} transforms into
\begin{align}
\mathbb{P}\big[\mathcal{O}(r_0,p_0)\big]&=\mathbb{P}\big[\max\limits_{j=1,\cdots,M}\mathrm{SIR}_j<2^{r_0}-1\big]\nonumber\\
&=\mathbb{P}\Big(\mathrm{SIR}_1<2^{r_0}-1,\cdots,  \mathrm{SIR}_M<2^{r_0}-1\Big)\nonumber\\
&\stackrel{(a)}{=}\mathbb{P}\Big(\mathrm{SIR}_j<2^{r_0}-1\Big)^M\nonumber\\
&\stackrel{(b)}{=}F_{\mathrm{SIR}_j}(2^{r_0}-1)^M,\label{errSC}
\end{align}
where $(a)$ follows from the fact that $\mathrm{SIR}_j$ is distributed independently on each antenna, and $(b)$ comes from using the definition of the CDF of $\mathrm{SIR}_j$.
\subsection{Maximal Radio Combining -- MRC}
Under the MRC scheme, \eqref{F} transforms into
\begin{align}
\mathbb{P}[\mathcal{O}(r_0,p_0)]&=\mathbb{P}\Big[\sum_{j=1}^{M}\mathrm{SIR}_j<2^{r_0}-1\Big]\nonumber\\
&=\mathbb{P}\big(\Psi<2^{r_0}-1\big)=F_{\Psi}(2^{r_0}-1),\label{Fpsi}
\end{align}
where $\Psi=\sum_{j=1}^{M}\mathrm{SIR}_j$. From Remark~\ref{Re1}, $\Psi$ can be represented as $\kappa\delta p_0 \sum_{j=1}^{M}\varphi_j$ where
\begin{align}
f_{\varphi_j}(x)=\kappa(1+x)^{-\kappa-1}. \label{pdfvarphi}
\end{align}
Thus,
\begin{align}
F_{\Psi}(2^{r_0}\!-\!1)\!=\!\mathbb{P}\Bigg[\sum_{j=1}^{M}\varphi_j\!<\!\frac{2^{r_0}\!-\!1}{\kappa\delta p_0}\Bigg]\!=\!F_\upsilon\Big(\frac{2^{r_0}\!-\!1}{\kappa\delta p_0}\Big),
\label{epmrc}
\end{align}
where $\upsilon=\sum_{i=1}^{M}\varphi_i$, while its CDF is given by \cite[Eq.(4.13)]{Zhang2.2017} 
\begin{align}
F_{\upsilon}(x)&=\int_{0}^{\infty}\frac{1}{u}(1-e^{-xu})e^{-Mu}\vartheta(M,\kappa,u)\mathrm{d}u,\label{cdfS}
\end{align}
with
\begin{align}
\vartheta(M,\kappa,u)&\!=\nonumber\\
\!\kappa^M\!\!\!\sum_{m=0}^{\lfloor\frac{ M\!-\!1}{2}\rfloor}&\!\!\binom{M}{2m\!+\!1}(-\pi^2)^m\xi(\kappa,u)^{M\!-\!2m\!-\!1}\Big(\frac{u^{\kappa}}{\kappa!}\Big)^{2m\!+\!1}\!,\\
\xi(\kappa,u)&\!=\!\frac{u^{\kappa}}{\kappa!}\!\Big(\sum_{m\!=\!1}^{\kappa}\!\frac{1}{m}\!-\!\varsigma\!-\!\ln(u)\Big)\!-\!\sum_{\substack{m\!=\!0\\ m\!\ne\!\kappa}}^{\infty}\!\frac{u^m}{(m\!-\!\kappa)m!},
\end{align}
and $\varsigma=0.5772156649...$ is the Eulers's constant. 

Unfortunately $F_{\upsilon}(x)$  is very difficult to evaluate, therefore, very time-consuming. In fact, it is also impossible to be evaluated for many combinations of parameter values $(M,\kappa,x)$, e.g, relatively small $x$ and relatively large $M$ and/or $\kappa$, for which calculation does not converge due to software/hardware limitations. Additionally, since $\mathbb{P}[\mathcal{O}(r_0,p_0)]=\varepsilon$ requires to be solved,  the inversion of $F_{\upsilon}(x)$ is needed, which is an even more cumbersome task. For those reasons, we provide next an accurate approximation for $F_{\upsilon}(x)$ in the left tail, and then we dedicate our attention to find $F_{\upsilon}^{-1}(\varepsilon)$.
\begin{theorem}\label{th3}
	The PDF and CDF of $\upsilon$ are approximated by
	\begin{align}
	f_{\upsilon}(x)&\!\approx\!\frac{\kappa^M M^{M\!-\!1}}{(M-1)!}\Big(1\!+\!\frac{x}{M}\Big)^{-1\!-\!M\kappa}\!\ln^{M-1}\Big(1\!+\!\frac{x}{M}\Big),\label{pdfG}\\
	F_{\upsilon}(x)&\!\approx 1-\frac{\Gamma\big(M,\kappa M\ln(1+x/M)\big)}{(M-1)!},\label{cdfG}
	\end{align}
	where \eqref{cdfG} converges to \eqref{cdfS} in the left tail.
\end{theorem}
\begin{proof}
	We have that 
	\begin{align}
	F_{\upsilon}(x)\!=\!\mathbb{P}\bigg[\sum_{j=1}^{M}\varphi_j\!<\!x\bigg]\!\stackrel{(a)}{=}\!\mathbb{P}\bigg[\sum_{j=1}^{M}\frac{(1\!+\!\varphi_j)}{M}\!<\!1\!+\!\frac{x}{M}\bigg],\label{dist}
	\end{align}
	where $(a)$ follows from adding $M$ and then dividing by $M$ on each side. The left term is the arithmetic mean of $1+\varphi_j,\ j=1,...,M$, thus, we are going to use again the relation between the arithmetic and geometric means. But first notice that according to \eqref{pdfvarphi} the mean of $\varphi_j$, $\bar{\varphi_j}=\frac{1}{\kappa-1},\ \forall\kappa>1$, decreases with $\kappa$ and already for $\kappa>2$ its value is below $1$, thus, $\varphi$ is expected to be smaller than $1$ with high probability when $\kappa$ is not too small.  Therefore, all results that comes next from using the geometric mean in the left term of \eqref{dist} are tighter when $\kappa$ increases and converge to the exact expression. But most importantly, the expressions converge in the left tail  where $x\rightarrow 0$, for which each of the summands $\varphi_j$ is expected to take much smaller values while getting far from $1$. We proceed as follows
	\begin{align}
	F_{\upsilon}(x)&\approx\mathbb{P}\bigg[\prod_{j=1}^{M}(1\!+\!\varphi_j)\!<\Big(\!1\!+\!\frac{x}{M}\Big)^M\bigg]\nonumber\\
	&= F_{\psi}\bigg(\Big(\!1\!+\!\frac{x}{M}\Big)^M\bigg),\label{cdfv}
	\end{align}
	where $\psi(M)=\prod_{j=1}^{M}\chi_j$, and $\chi_j=1+\varphi_j\sim \mathcal{P}(\kappa,1)$ with PDF
	\begin{align}
	f_{\chi_j}(x)=\kappa x^{-\kappa-1},\ x\ge 1. \label{chi}
	\end{align}
	Now we are going to prove by induction that the PDF of $\psi$ is given by
	\begin{align}
	f_{\psi}(x|M)=\frac{\kappa^M}{(M-1)!}\ln^{M-1}(x)x^{-\kappa-1},\  x\ge 1. \label{genM}
	\end{align}
	The proof proceed as follows.
	\begin{itemize}
		\item For $M=2$ we have that $\psi(2)=\chi_1\chi_2$, thus
		\begin{align}
		f_{\psi}(x|2)&\!\stackrel{(a)}{=}\!\int_{1}^{x}\frac{1}{\chi_1}f_{\chi_1}(\chi_1)f_{\chi_2}(x/\chi_1)\mathrm{d}\chi_1\!\nonumber\\
		&\stackrel{(b)}{=}\!\kappa^2x^{-\kappa-1}\!\int_{1}^{x}\frac{1}{\chi_1}\mathrm{d}\chi_1\!\stackrel{(c)}{=}\!\kappa^2x^{-\kappa-1}\ln x,\label{m2}
		\end{align}
		where $(a)$ comes from the distribution of the product of two random variables, $(b)$ comes from substituting \eqref{chi}, and $(c)$ follows from solving the integral. Notice that \eqref{m2} matches \eqref{genM} for $M=2$.
		\item Assume now that \eqref{genM} holds for a given $M\ge 1$ and we are going to check whether it also holds for $M+1$. In this case we have that $\psi(M+1)=\chi_{M+1}\psi(M)$, thus,
		\begin{align}
		f_{\psi}(x|M+1)&=\int_{1}^{x}\frac{1}{\chi}f_{\chi}(\chi)f_{\psi}(x/\chi|M)\mathrm{d}\chi\nonumber\\
		&\stackrel{(a)}{=}\frac{\kappa^{M+1}}{(M-1)!}x^{-\kappa-1}\int_{1}^{x}\frac{1}{\chi}\ln^{M-1}\frac{x}{\chi}\mathrm{d}\chi\nonumber\\
		&\stackrel{(b)}{=}\frac{\kappa^{M+1}\ln^{M}(x)x^{-\kappa-1}}{(M-1)!M},\label{pr}
		\end{align}
		where $(a)$ follows from substituting \eqref{chi}, \eqref{genM} and simple algebraic simplifications, while $(b)$ comes from solving the integral. By using $(M-1)!M=M!$ notice that \eqref{pr} matches \eqref{genM} with $M\leftarrow M+1$. 
	\end{itemize} 
	Therefore, \eqref{genM} holds. Now, the CDF of $\psi$ is given by
	\begin{align}
	F_{\psi}(x|M)&\!=\!\int_{1}^{x}\!\!f_{\psi}(u|M)\mathrm{d}u\!=\!\frac{\kappa^M}{(M\!-\!1)!}\int_{1}^{x}\!\!u^{-\kappa-1}\ln^{M\!-\!1}u\mathrm{d}u\nonumber\\
	&=-\frac{\Gamma(M,\kappa\ln u)}{(M-1)!}\Bigg|_1^x=1-\frac{\Gamma(M,\kappa \ln x)}{(M-1)!}.\label{cdf}
	\end{align}
	Substituting \eqref{cdf} into \eqref{cdfv} we attain \eqref{cdfG}, while \eqref{pdfG} comes from evaluating $f_{\upsilon}(x)\!=\!\frac{d}{d x}F_{\upsilon}(x)$.
\end{proof}
\begin{figure}[t!]
	\centering
	\subfigure{\includegraphics[width=0.45\textwidth]{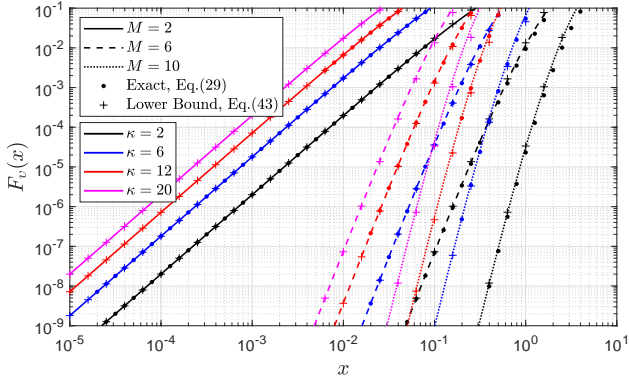}}
	\vspace*{-2mm}
	\caption{Left tail of $F_{\upsilon}(x)$. Comparison between the exact \eqref{cdfS}, approximate \eqref{cdfG} (represented with the lines), and lower bound \eqref{app} with $\ln(1+x/M)=x/M$, expressions.}
	\vspace*{-2mm}		
	\label{Fig4}
\end{figure}
Fig.~\ref{Fig4} shows the incredible accuracy of \eqref{cdfG} in the left tail. Only a slight divergence from the exact expression is observable when $\kappa$ is relatively small, e.g., $\kappa=2$, at the same time that the reliability is not too restrictive, $F_{\upsilon}(x)\ge 10^{-2}$. This is in-line with the arguments we used when proving Theorem~\ref{th3}. Using expressions \eqref{pdfG} and \eqref{cdfG} is twofold advantageous: i) they are  relatively easy to evaluate and ii) they can be evaluated in regions where the exact expressions cannot. \footnote{Regarding this last aspect, notice that \eqref{cdfS} does not converge for $\kappa=20$ and also for $\kappa=12,\ M=10,\ F_{\upsilon}(x)\ge 10^{-4}$, just for mentioning two examples.}

Although an easy-to-evaluate expression for $F_{\upsilon}(x)$ was given in \eqref{cdfG}, it is not analytically invertible, thus, $F_{\upsilon}^{-1}(\varepsilon)$ requires to be computed numerically\footnote{Note that there software packages to evaluate the inverse gamma function, e.g., \texttt{gammaincinv} in MatLab and \texttt{InverseGammaRegularized} in Wolfram Mathematica.}. Following result aims at alleviating this issue.
\begin{lemma}
	$F_{\upsilon}^{-1}(\varepsilon)$ approximates to
	\begin{align}
	\frac{(M!)^{1/M}}{\kappa}\big|\ln\big(1-\varepsilon^{1/M}\big)\big|,\label{eqEth}
	\end{align}
	specially when $\varepsilon$ is very restrictive and $M$ is not too large.
\end{lemma}
\begin{proof}
	According to \cite[Eq. (8.10.11)]{Thompson.2011} we have that
	\begin{align}
	F_{\upsilon}(x)\ge \Big(1-e^{-(M!)^{-1/M}\kappa M\ln\big(1+x/M\big)}\Big)^M,\label{app}
	\end{align}
	where equality holds for $M=1$ and diverges slowly when $M$ increases. Additionally, this lower bound is very tight in the left tail of the curve, e.g., when $\varepsilon$ is more restrictive. We require to isolate $x$ from $F_{\upsilon}(x)=\varepsilon$, and notice that for $\varepsilon\rightarrow 0$ we have $x\rightarrow 0$, thus, we can take $\ln(1+x/M)\lesssim x/M$, which makes \eqref{app} even more accurate when $\varepsilon$ is not too small. The tightness of the lower bound is clearly shown in Fig.~\ref{Fig5}. Finally we attain \eqref{eqEth} straightforwardly. 
\end{proof}
\section{Optimum Joint Power Control and Rate Allocation}\label{joint}
As highlighted at the end of Subsection~\ref{C}, the optimum resource allocation lies on the curve $\mathbb{P}[\mathcal{O}(r_0,p_0)]=\varepsilon$. Specifically, for SC and SSC and based on \eqref{errSC}, the exact relation between $r_0$ and $p_0$ is given by $\prod_{i=1}^{\kappa}\Big(1+(2^{r_0}-1) \frac{p_i\lambda_i}{p_0\lambda_0}\Big)=\frac{1}{1-\varepsilon^{1/M}}$, while for MRC we were unable to find it. Notice that using such exact intricate relation, even more intricate for large $\kappa$, is additionally not advisable because the solution pairs are expected to depend on $\lambda_0$ and each $p_i\lambda_i$ separately, which is not suitable since such information is difficult to obtain for $T_0$. Following result aims at addressing these issues by providing a relatively simple, yet practically useful, relation between  $r_0$ and $p_0$ for all the diversity schemes.
\begin{lemma}
When $\varepsilon\ll 1$ the curve $\mathbb{P}[\mathcal{O}(r_0,p_0)]=\varepsilon$ is tightly approximated by
\begin{align}
r_0\approx \log_2(\omega p_0+1),\label{eqSCap}
\end{align}
where
\begin{align}\label{omega}
\omega\!=\!\left\{ \begin{array}{ll}
\!\!\kappa\delta \Big(\big(1-\varepsilon^{1/M}\big)^{-\frac{1}{\kappa}}-1\Big), &   \mathrm{for\ SC\ and\ SSC} \\
\!\!\delta  (M!)^{1/M} \big|\ln\big(1\!-\!\varepsilon^{1/M}\big)\big|, & \mathrm{for\ MRC}
\end{array}
\right.\!.
\end{align}
\end{lemma}
\begin{proof}
	For SC and SSC we can compute  $F_{f(\mathbf{SIR})}^{-1}(\varepsilon)$ accurately for $\varepsilon\ll 1$ by using 	
	\eqref{errSC} with \eqref{Fsnrapp}. Meanwhile, for MRC an accurate approximation is given by $F_{f(\mathbf{SIR})}^{-1}(\varepsilon)=\kappa\delta p_0 F_v^{-1}(\varepsilon)$, where $F_v^{-1}(\varepsilon)$ is given by \eqref{eqEth}. Substituting $F_{f(\mathbf{SIR})}^{-1}(\varepsilon)$ into \eqref{Finv} yields \eqref{eqSCap}.
\end{proof}
Again, notice that the significance of \eqref{eqSCap} is undeniable since it shows that rather than depending on each $p_i\lambda_i$ separately, $r_0$ ultimately depends on the number of interfering nodes, the number of receive antennas, the reliability constraint and the ratio between the average signal and average interference powers, which are easy/viable to estimate/know. Now we are in condition to make the following proposition.
\begin{lemma}\label{pro1}
	Solving $\mathbf{P1}$ is equivalent to solve
\begin{subequations}\label{P2}
	\begin{alignat}{2}
	\mathbf{P2:}\ &\argmax{x}       &\ & \ln\mathrm{EE}\big(\log_2( e^x\omega+1
	),e^x,\beta\big) \label{P2:a}\\
	&\mathrm{s.t.} &    & \ln\max\!\Big(p_{\min},\frac{2^{r_{\min}}\!-\!1}{\omega}\Big)\!\le\! x\! \le\! \ln p_{\max} \label{P2:b}
	\end{alignat}
\end{subequations}
	for $\varepsilon\ll 1$ and $p_0^*=e^{x^*}$.
\end{lemma}
\begin{proof}
	It is required that $\log_2(\omega p_0+1)\ge r_{\min}\rightarrow p_0\ge \frac{2^{r_{\min}}-1}{\omega}$ and $p_{\min} \le p_0\le p_{\max}$ according to \eqref{eq4:c} and \eqref{eq4:d}, respectively, and combining them yields $\max\big(\frac{2^{r_{\min}}-1}{\omega},p_{\min}\big) \le p_0\le p_{\max}$. The objective function can be written now as a function of $p_0$. Since the resultant objective function is not concave we use the fact that optimizing it conduces to the same result as optimizing $\ln \mathrm{EE}$ and the optimization over $p_0$ is equivalent to optimize over $x=\ln p_0$. Hence, such transformation yields \textbf{P2}.
\end{proof}
\begin{theorem}\label{the2}
	Setting
	\begin{subequations}
		\begin{alignat}{2}
	\rho&=\frac{\frac{\omega\eta(p_t+M^{\beta}p_r+p_{syn})-1}{\mathcal{W}\big(\frac{\omega\eta(p_t+M^{\beta}p_r+p_{syn})-1}{e}\big)}-1}{\omega},\\
    \varrho&=\frac{1}{\ln 2}\Big(\mathcal{W}\Big(\frac{\omega\eta(p_t+M^{\beta}p_r+p_{syn})-1}{e}\Big)+1\Big),
 		\end{alignat}
	\end{subequations}
	the optimum resource allocation is given by
	\begin{itemize}
		\item If $\rho< \ln\max\Big(p_{\min},\frac{2^{r_{\min}}\!-\!1}{\omega}\Big)$ then
			\begin{subequations}
				\begin{alignat}{2}
		p_0^*&=\max\Big(p_{\min},\frac{2^{r_{\min}}\!-\!1}{\omega}\Big),\\ 		r_0^*&=\max\big(\log_2(\omega p_{\min}+1),r_{\min}\big);
			\end{alignat}
			\end{subequations}
		\item If $ \ln\max\Big(p_{\min},\frac{2^{r_{\min}}\!-\!1}{\omega}\Big)\le \rho\le \ln p_{\max} $ then
		\begin{align}
		p_0^*&=\rho, &
		r_0^*&=\varrho;
		\end{align}
		\item If $ \rho> \ln p_{\max} $ then
		\begin{align}
		p_0^*&=p_{\max}, & r_0^*&=\log_2(\omega p_{\max}+1).
		\end{align}
	\end{itemize}	
\end{theorem}
\begin{proof}
	Let us denote $g(x)$ as the objective function of \textbf{P2} and based on \eqref{P2:a} and \eqref{EE}, it is given by
	\begin{align}
	g(x)&\!=\!\ln\frac{\ln (e^x \omega+1)}{e^x/\eta\!+\!p_t\!+\!M^{\beta}p_r\!+\!p_{syn}}\nonumber\\
	&\!=\!\ln \big(\ln\! (e^x \omega\!+\!1)\big)\!-\!\ln\big(\tfrac{e^x}{\eta}\!+\!p_t\!+\!M^{\beta}p_r\!+\!p_{syn}\big).
	\end{align}
	Notice that we have ignored the term $1-\mathbb{P}[\mathcal{O}(r_0,p_0)]$ since by design it is equal to $1-\varepsilon\approx 1$, and we have used $\ln(\cdot)$ instead of $\log_2(\cdot)$, which does not affect the optimization of $g(x)$. Now, the first and second derivatives of $g(x)$ are
	\begin{align}
	&\frac{d}{d x}g(x)\!=\!\frac{\omega e^x}{(e^x\omega\!+\!1)\ln(e^x\omega\!+\!1)}\!-\!\frac{e^x/\eta}{e^x/\eta\!+\!p_t\!+\!M^{\beta}p_r\!+\!p_{syn}}\nonumber\\
	&\!=\!\frac{\omega(e^x\!+\!\eta(p_t\!+\!M^{\beta}p_r\!+\!p_{syn}))\!-\!(1\!+\!e^x\omega)\ln(e^x\omega\!+\!1)}{(e^x\!+\!\eta(p_t\!+\!M^{\beta}p_r\!+\!p_{syn}))(1\!+\! e^x\omega)\ln(e^x\omega\!+\!1)e^{-x}}\!,\label{d1}\\
	&\frac{d^2}{d x^2}g(x)=\nonumber\\
	&\!-\!\frac{e^x\omega\big(e^x\omega\!-\!\ln(e^x\omega\!+\!1)\big)}{(e^x\omega\!+\!1)^2\ln^2(1\!+\!e^x\omega)}\!-\!\frac{e^x(p_t\!+\!M^{\beta}p_r\!+\!p_{syn})}{\frac{e^x}{\eta}\!+\!(p_t\!+\!M^{\beta}p_r\!+\!p_{syn})^2},\label{d2}
	\end{align}
	where the second derivative comes from taking the derivative of $\frac{d}{dx}g(x)$ in the first line of \eqref{d1}. Notice that $\frac{d^2}{dx^2}g(x)<0,\ \forall x\in \mathcal{R}$ since $e^x\omega\ge \ln(e^x\omega+1)$, thus, $g(x)$ is concave on $x$ and it has a global maximum on the solution of $\frac{d}{d x}g(x)=0$ which is obtained as follows
	\begin{align}	
	\omega(e^x\!+\!\eta(p_t\!+\!M^{\beta}p_r\!+\!p_{syn}))&\!-\!(1\!+\! e^x\omega)\ln(e^x\omega\!+\!1)\!=\!0\nonumber\\
	\omega\eta(p_t\!+\!M^{\beta}p_r\!+\!p_{syn})\!-\!1\!&=\!(1\!+\! e^x\omega)\big(\ln(1\!+\!e^x\omega)\!-\!1\big)\nonumber\\
    \frac{\omega\eta(p_t\!+\!M^{\beta}p_r\!+\!p_{syn})\!-\!1}{e}\!&=\!e^{\ln(1\!+\! e^x\omega)\!-\!1}\!\big(\!\ln(1\!+\!e^x\omega)\!-\!1\!\big)\nonumber\\
	\mathcal{W}\!\Big(\!\frac{\omega\eta(p_t\!+\!M^{\beta}p_r\!+\!p_{syn})\!-\!1}{e}\!\Big)\!&\stackrel{(a)}{=}\!\ln(1\!+\!e^x\omega)\!-\!1\nonumber\\
	\ln\!\frac{e^{\mathcal{W}\big(\frac{\omega\eta(p_t\!+\!M^{\beta}p_r\!+\!p_{syn})\!-\!1}{e}\big)}e\!-\!1}{\omega}\!&\stackrel{(b)}{=}x^*\nonumber\\
	\ln\!\frac{\frac{\omega\eta(p_t\!+\!M^{\beta}p_r\!+\!p_{syn})\!-\!1}{\mathcal{W}\big(\frac{\omega\eta(p_t\!+\!M^{\beta}p_r\!+\!p_{syn})\!-\!1}{e}\big)}\!-\!1}{\omega}&\stackrel{(c)}{=}x^*,\label{optx}
	\end{align}
	where $(a)$ comes from the definition of the Lambert $W$ function, specifically its main branch since $\ln(e^x\omega+1)-1\ge -1$, which guarantees finding the appropriate real solution \cite{Corless.1996}; $(b)$ follows from isolating $x$; and $(c)$ from using the property  $e^{\mathcal{W}(a)}=a/\mathcal{W}(a)$. Notice that $x^*$ in \eqref{optx} is the solution of \textbf{P2} as long as it is in the interval specified by \eqref{P2:b}; otherwise, if it is greater than $\ln p_{\max}$ then $x^*=\ln p_{\max}$, while if it is smaller than $\ln\max\Big(p_{\min},\frac{2^{r_{\min}}-1}{\omega}\Big)$ then $x^*=\ln\max\Big(p_{\min},\frac{2^{r_{\min}}-1}{\omega}\Big)$. Now, returning to original variables by using $r_0^*=\log_2(\omega p_0^*+1)$ and $p_0^*=e^{x^*}$ we attain the resource allocation scheme given in Theorem~\ref{the2}. Notice that to obtain $r_0^*$ in the simplified form it is required the use of property $\ln\mathcal{W}(a)=\ln a-\mathcal{W}(a)$ for $a> 0$.
\end{proof}

Notice that, as commented at the end of Subsection~\ref{C}, the solution is feasible only when $\mathbb{P}[\mathcal{O}(r_{\min},p_{\max})]\le \varepsilon$, thus, according to \eqref{eqSCap} it is only necessary to check that $r_{\min}\le  \log_2\Big(\kappa\delta p_{\max}\Big(\big(1-\varepsilon^{1/M}\big)^{-\frac{1}{\kappa}}-1\Big)+1
\Big)$.
\begin{figure}[t!]
	\centering
	\subfigure{\includegraphics[width=0.45\textwidth]{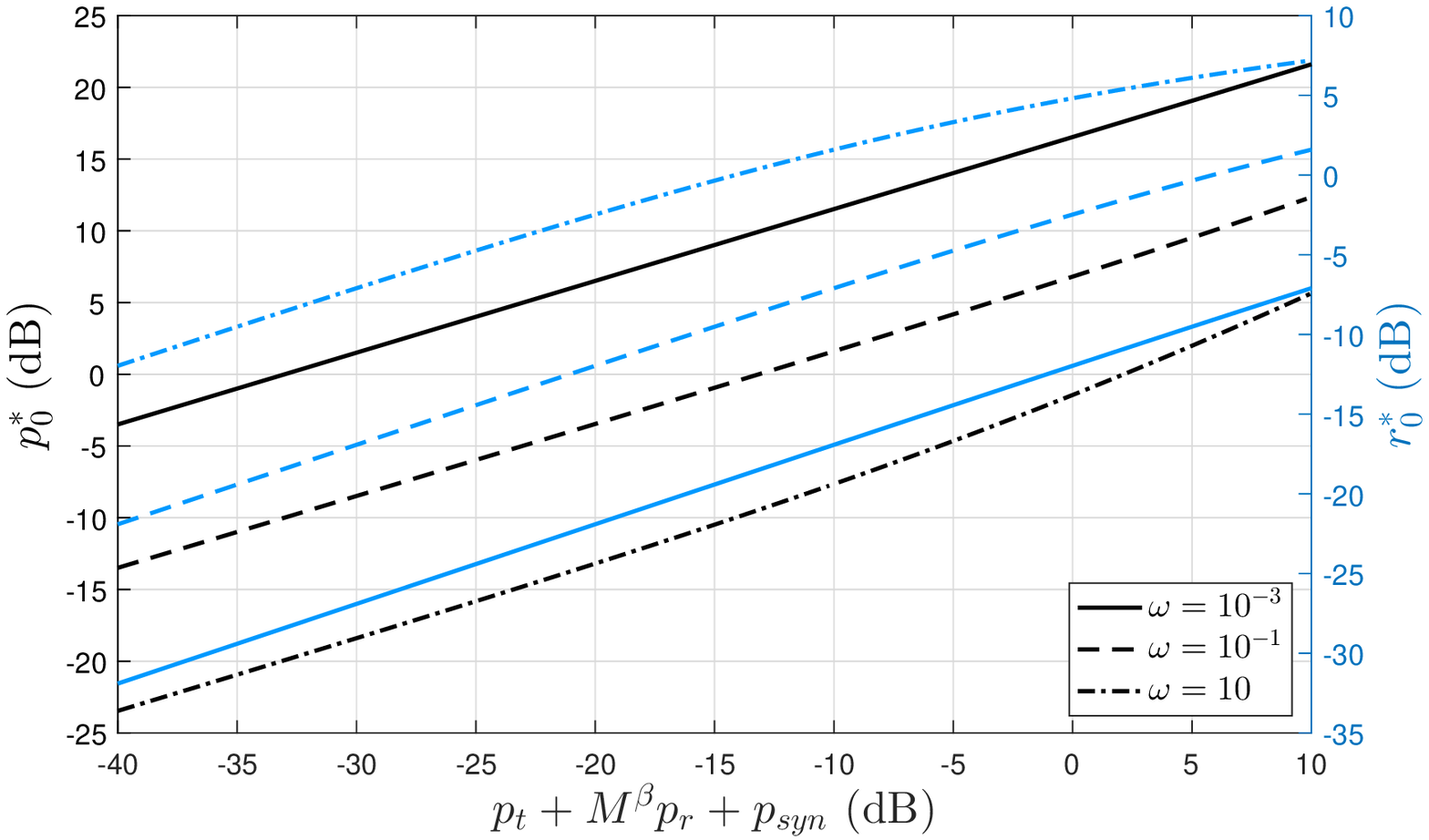}}\\
	\subfigure{\includegraphics[width=0.45\textwidth]{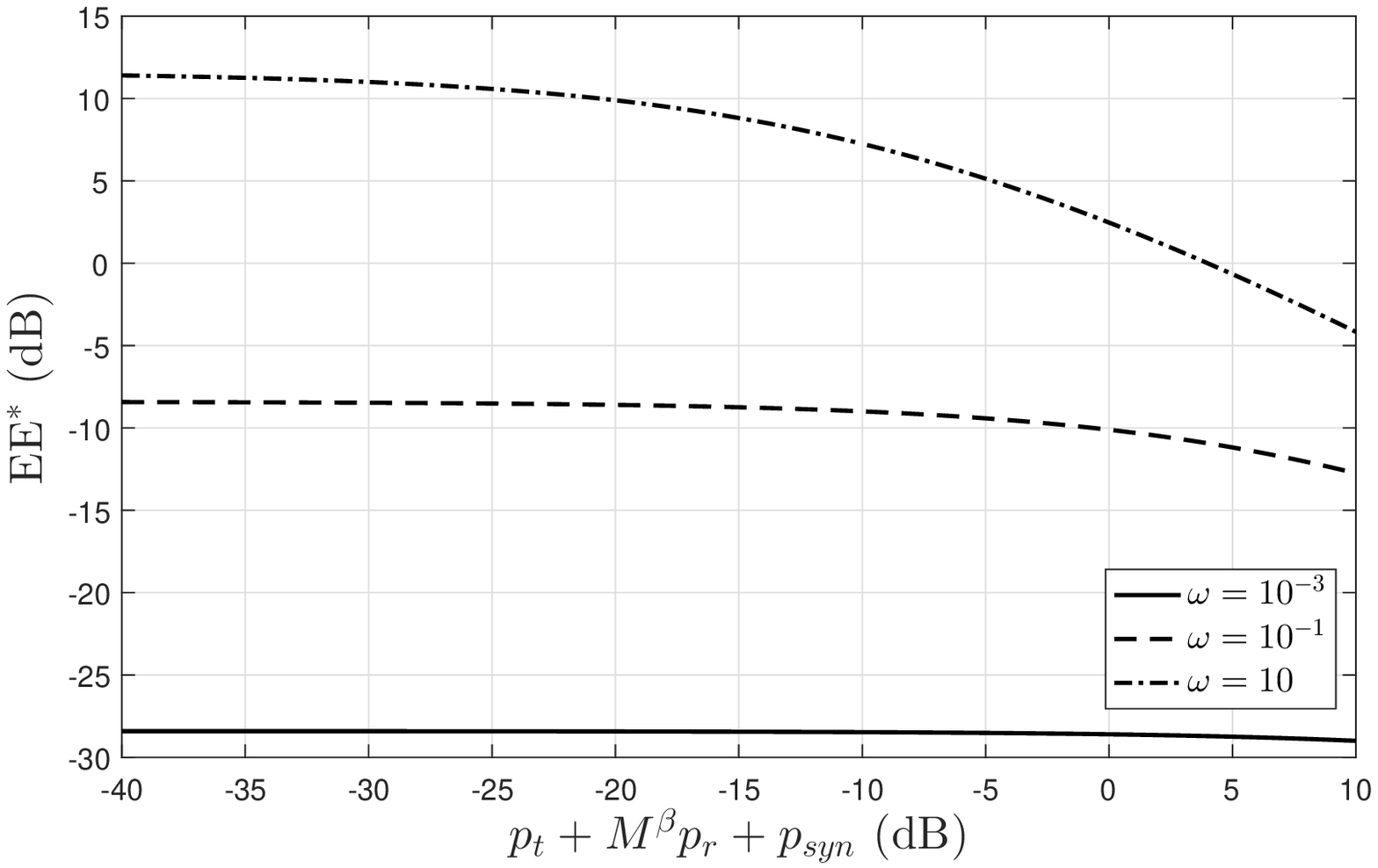}}
	\vspace*{-2mm}
	\caption{Optimum performance as a function of $p_t+M^{\beta}p_r+p_{syn}$ for $\omega\in\{10^{-3},10^{-1},10\}$ and $\eta=1$. $(a)$ $p_0^*$ in black and $r_0^*$ in blue (left), and $(b)$ $\mathrm{EE}^*$ (right).}	
	\vspace*{-2mm}	
	\label{Fig5}
\end{figure}

Fig.~\ref{Fig5} shows the dependence of $p_0^*,\ r_0^*$ and $\mathrm{EE}^*$ on $p_t+M^{\beta}p_r+p_{syn}$ and $\omega$ for unconstrained transmit power and rate, e.g., $r_{\min}=0$, $p_{\min}=0$ and $p_{\max}=\infty$. The greater the fixed power consumption figure, the greater the optimum transmit power and transmit rate, while the optimum energy it is only affected at relatively large $\omega$. Notice that $p_0^*$ decreases with $\omega$, while $r_0^*$ and $\mathrm{EE}^*$ are increasing functions of $\omega$. According to \eqref{omega} $\omega$ is an increasing function of $\varepsilon$, thus, $p_0^*$ increases as the system reliability increases while the optimum rate decreases, as well as the energy efficiency, but at slow pace.
\subsection{SC vs MRC}\label{SubC}
\begin{lemma}\label{cor3}
	Following relations hold
	\begin{align}
	p_{0,\mathrm{mrc}}^* < p_{0,\mathrm{sc}}^*,& &r_{0,\mathrm{mrc}}^*> r_{0,\mathrm{sc}}^*.\label{poro}
	\end{align}
	Thus, $\mathrm{EE}^*_{\mathrm{mrc}}>\mathrm{EE}^*_{\mathrm{sc}}$.
\end{lemma}
\begin{proof}
	According to \eqref{omega} we have that $\omega=0$ when $\varepsilon=0$. Additionally, $\omega$  is an increasing function of $z=\varepsilon^{1/M}$ since
	\begin{subequations}
		\begin{alignat}{2}		
	\frac{d}{d z}\omega_{\mathrm{sc,ssc}}&=\frac{1}{(1-z)^{1/\kappa+1}}>0,\\ 
	\frac{d}{d z}\omega_{mrc}&=\frac{(M!)^{1/M}}{1-z}>0;
	\end{alignat}
	\end{subequations}
	but $\omega_{mrc}$ grows faster than $\omega_{\mathrm{sc,ssc}}$ because
	\begin{align}
	\frac{d}{dz}\omega_{mrc} &>  \frac{d}{d z}\omega_{\mathrm{sc,ssc}}\nonumber \\
	\frac{(M!)^{1/M}}{1-z} &>\frac{1}{(1-z)^{1/\kappa+1}}\nonumber\\
	(M!)^{1/M} &>\frac{1}{(1-z)^{1/\kappa}},
	\end{align}
	where the last condition is always satisfied since $z<1-(M!)^{-\kappa/M}$ for $\kappa\ge 1$ and practical reliability constraints, e.g., $\varepsilon<10^{-1}$. Above implies that $\omega_{\mathrm{mrc}}>\omega_{\mathrm{sc,ssc}}$, thus, according to the discussion related to Fig.~\ref{Fig5}a and that $p_0^*$ and $r_0^*$ share the same dependence on $\omega$ for SC and MRC schemes\footnote{The dependence is shown in Theorem~\ref{the2} and notice that is strictly the same for SC and MRC since $\beta=1$ for both schemes.}, \eqref{poro} holds.
\end{proof}
\begin{figure}[t!]
	\centering
	\subfigure{\includegraphics[width=0.45\textwidth]{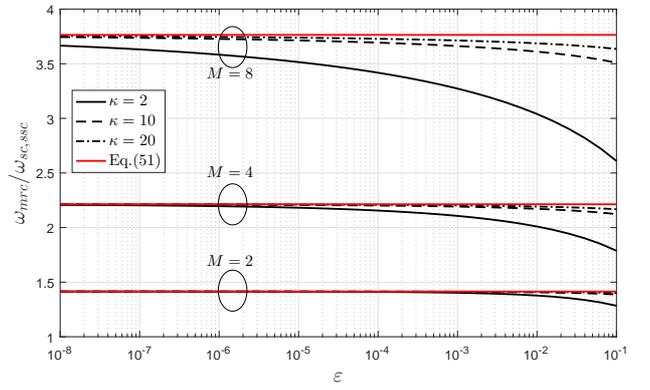}}
	\vspace*{-2mm}
	\caption{Gap between $\omega_{\mathrm{mrc}}$ and $\omega_{\mathrm{sc,ssc}}$ as a function of $\varepsilon$ for $M\in\{2,4,8\}$ and $\kappa\in\{2,10,20\}$.}		
	\vspace*{-2mm}
	\label{Fig6}	
\end{figure}

In the ultra-reliability regime, the asymptotic gap between $\omega_{\mathrm{sc,ssc}}$ and $\omega_{\mathrm{mrc}}$ can be calculated as follows
\begin{align}
\lim\limits_{\varepsilon\rightarrow 0}\frac{\omega_{\mathrm{mrc}}}{\omega_{\mathrm{sc,ssc}}}&\stackrel{(a)}{=}\lim\limits_{\varepsilon\rightarrow 0}\frac{d\omega_{\mathrm{mrc}}/d\varepsilon}{d\omega_{\mathrm{sc,ssc}}/d\varepsilon}\nonumber\\
&\stackrel{(b)}{=}\lim\limits_{\varepsilon\rightarrow 0}\frac{(M!)^{1/M}(1-\varepsilon^{1/M})^{-1}\frac{1}{M}\varepsilon^{1/M-1}}{\frac{1}{M}(1-\varepsilon^{1/M})^{-1/\kappa-1}\varepsilon^{1/M-1}}\nonumber\\
&=\lim\limits_{\varepsilon\rightarrow 0}(M!)^{1/M}(1\!-\!\varepsilon^{1/M})^{1/\kappa}\!=\!(M!)^{1/M},\label{l1}
\end{align}
where $(a)$ comes from using L'H\^opital's rule and $(b)$ follows from taking the derivative of  \eqref{omega} with respect to $\varepsilon$. Thus, the asymptotic gap is only function (an increasing function) of the number of antennas. This is illustrated in Fig.~\ref{Fig6}, where we can also check that the non-asymptotic gap narrows as the reliability constraint relaxes and the number of interfering transmitters increases.

More important than the relation between $\omega_{\mathrm{mrc}}$ and $\omega_{\mathrm{sc}}$, is the relation between $r^*_{0,\mathrm{sc}}$ and $r^*_{0,\mathrm{mrc}}$, and between $p^*_{0,\mathrm{sc}}$ and $p^*_{0,\mathrm{mrc}}$. Therefore, from Theorem~\ref{the2}, when assuming no constraints in the power and the rate, we have
\begin{align}
\lim\limits_{\varepsilon\rightarrow 0}\frac{r^*_{0,\mathrm{mrc}}}{r^*_{0,\mathrm{sc}}}&=\lim\limits_{\varepsilon\rightarrow 0}\frac{\mathcal{W}\Big(\frac{\omega_{\mathrm{mrc}} \eta (p_t+Mp_r+p_{syn})-1}{e}\Big)+1}{\mathcal{W}\Big(\frac{\omega_{\mathrm{sc}} \eta (p_t+Mp_r+p_{syn})-1}{e}\Big)+1}\nonumber\\
&\stackrel{(a)}{=}\lim\limits_{\omega_{\mathrm{sc}}\rightarrow 0}\frac{\mathcal{W}\Big(\frac{(M!)^{1/M}\omega_{\mathrm{sc}} \eta (p_t+Mp_r+p_{syn})-1}{e}\Big)\!+\!1}{\mathcal{W}\Big(\frac{\omega_{\mathrm{sc}} \eta (p_t+Mp_r+p_{syn})-1}{e}\Big)+1}\nonumber\\
&=(M!)^{1/(2M)},\label{l2}
\end{align}
where $(a)$ comes from using $\omega_{\mathrm{mrc}}=(M!)^{1/M}\omega_{\mathrm{sc}}$ which holds when $\varepsilon\rightarrow 0$ according to \eqref{l1}, and from the fact that as $\varepsilon\rightarrow 0$ we have that $\omega_{\mathrm{sc}}\rightarrow 0$. Similarly, when analyzing the asymptotic gap in the optimum transmit power, yields
\begin{align}
\lim\limits_{\varepsilon\rightarrow 0}\frac{p^*_{0,\mathrm{mrc}}}{p^*_{0,\mathrm{sc}}}&=(M!)^{-1/(2M)}.\label{l3}
\end{align}
\subsection{SC vs SSC}\label{B1}
SSC is more energy efficient than SC since it is able of achieving the same reliability performance with reduced power consumption as shown in \eqref{EE}. 
Additionally, only one of the following alternatives holds for guaranteeing $\mathbb{P}\big[\mathcal{O}(r_{\mathrm{0,ssc}},p_{\mathrm{0,ssc}})\big]=\varepsilon$ as discussed in Fig.~\ref{Fig2}:  i) $r_{\mathrm{0,ssc}}^*>r_{\mathrm{0,sc}}^*$, $p_{\mathrm{0,ssc}}^*>p_{\mathrm{0,sc}}^*$, or ii) $r_{\mathrm{0,ssc}}^*<r_{0,\mathrm{sc}}^*$, $p_{0,\mathrm{ssc}}^*<p_{0,\mathrm{sc}}^*$. According to the results and discussions around Fig.\ref{Fig5}, $p_0^*$ and $r_0^*$ increase with the circuitry power consumption, hence case ii) holds. Summarizing:
\begin{align}
p_{0,\mathrm{ssc}}^* < p_{0,\mathrm{sc}}^*,& &r_{0,\mathrm{ssc}}^*< r_{0,\mathrm{sc}}^*, & & \mathrm{EE}_{\mathrm{ssc}}^*>\mathrm{EE}_{\mathrm{sc}}^*.\label{porossc}
\end{align}
Now, let's proceed with an analytical characterization of the performance gap between these two diversity schemes. Since $\omega_{\mathrm{ssc}}=\omega_{\mathrm{sc}}=\omega_{\mathrm{sc,ssc}}$ and focusing on the ultra-reliability regime with no constraints in the power and rate, we have that
\begin{align}
\lim_{\varepsilon\rightarrow 0}\frac{r_{0,\mathrm{ssc}}^*}{r_{0,\mathrm{sc}}^*}&\!=\!\lim_{\varepsilon\rightarrow 0}\frac{\mathcal{W}\Big(\frac{\omega_{\mathrm{sc,ssc}}\eta(p_t\!+\!p_r\!+\!p_{syn})-1}{e}\Big)\!+\!1}{\mathcal{W}\Big(\frac{\omega_{\mathrm{sc,ssc}}\eta(p_t\!+\!Mp_r\!+\!p_{syn})\!-\!1}{e}\Big)\!+\!1}\nonumber\\
=\!\lim_{\omega_{\mathrm{sc,ssc}}\rightarrow 0}&\frac{\mathcal{W}\Big(\frac{\omega_{\mathrm{sc,ssc}}\eta(p_t\!+\!Mp_r\!+\!p_{syn})\frac{p_t\!+\!p_r\!+\!p_{syn}}{p_t\!+\!Mp_r\!+\!p_{syn}}\!-\!1}{e}\Big)\!+\!1}{\mathcal{W}\Big(\frac{\omega_{\mathrm{sc,ssc}}\eta(p_t\!+\!Mp_r\!+\!p_{syn})\!-\!1}{e}\Big)\!+\!1}\nonumber\\
&\!=\!\sqrt{\frac{p_t\!+\!p_r\!+\!p_{syn}}{p_t\!+\!Mp_r\!+\!p_{syn}}},\label{ga}
\end{align}
which is smaller than $1$ for every $M>1,\ p_r>0$. By doing similarly when analyzing the asymptotic gap in the optimum transmit power we attain the same result as in \eqref{ga}
\begin{align}
\lim_{\varepsilon\rightarrow 0}\frac{p_{0,\mathrm{ssc}}^*}{p_{0,\mathrm{sc}}^*}=\sqrt{\frac{p_t+p_r+p_{syn}}{p_t+Mp_r+p_{syn}}}.\label{ga2}
\end{align}
\subsection{SSC vs MRC}\label{B2}
As discussed in Subsection~\ref{SubC} $\omega_{\mathrm{mrc}}$ grows faster than $\omega_{\mathrm{sc,ssc}}$, however $p_0^*$ and $r_0^*$ no longer share the same dependence on $\omega$ for SSC and MRC schemes, hence, the same arguments can not be applied. Intuitively, $\mathrm{EE}_{\mathrm{mrc}}^*>\mathrm{EE}_{\mathrm{ssc}}^*$ should hold for relatively small $M$, however as $M$ increases the situation is reversed since the power consumption soars and for relatively large $M$ $\mathrm{EE}_{\mathrm{mrc}}^*<\mathrm{EE}_{\mathrm{ssc}}^*$ should hold. Characterizing analytically such trade-off is cumbersome, however if we limit our discussion to the ultra-reliability regime where $\varepsilon\rightarrow 0$ we are able to provide valuable insights as we do next.

Notice that
	\begin{align}
	\lim_{\varepsilon\rightarrow 0}\frac{r_{0,\mathrm{mrc}}^*}{r_{0,\mathrm{ssc}}^*}&=\lim_{\varepsilon\rightarrow 0}\frac{r_{0,\mathrm{mrc}}^*/r_{0,\mathrm{sc}}^*}{r_{0,\mathrm{ssc}}^*/r_{0,\mathrm{sc}}^*}=\frac{\lim_{\varepsilon\rightarrow 0}r_{0,\mathrm{mrc}}^*/r_{0,\mathrm{sc}}^*}{\lim_{\varepsilon\rightarrow 0}r_{0,\mathrm{ssc}}^*/r_{0,\mathrm{sc}}^*}\nonumber\\
	&=\sqrt{(M!)^{1/M}\frac{p_t+Mp_r+p_{syn}}{p_t+p_r+p_{syn}}},\label{limr}\\
	\lim_{\varepsilon\rightarrow 0}\frac{p_{0,\mathrm{mrc}}^*}{p_{0,\mathrm{ssc}}^*}&=\lim_{\varepsilon\rightarrow 0}\frac{p_{0,\mathrm{mrc}}^*/p_{0,\mathrm{sc}}^*}{p_{0,\mathrm{ssc}}^*/p_{0,\mathrm{sc}}^*}=\frac{\lim_{\varepsilon\rightarrow 0}p_{0,\mathrm{mrc}}^*/p_{0,\mathrm{sc}}^*}{\lim_{\varepsilon\rightarrow 0}p_{0,\mathrm{ssc}}^*/p_{0,\mathrm{sc}}^*}\nonumber\\
	&=\sqrt{(M!)^{-1/M}\frac{p_t+Mp_r+p_{syn}}{p_t+p_r+p_{syn}}},\label{limp}
	\end{align}
	where the last step in \eqref{limr} and \eqref{limp} comes from using \eqref{l2} and \eqref{ga}, and \eqref{l2} and \eqref{ga2}, respectively.
	\begin{theorem}
		As $\varepsilon\rightarrow 0$ and with no constraints on the power and rate, $p_{0,\mathrm{ssc}}^*\le p_{0,\mathrm{mrc}}^*$ when
		\begin{align}
		M\gtreqless \frac{0.786p_r-0.214(p_t+p_{syn})}{0.626p_r-0.374(p_t+p_{syn})}, \label{theS}
		\end{align}
		for $0.626p_r-0.374(p_t+p_{syn})\gtreqless 0$.
	\end{theorem}
\begin{proof}
	We require to solve $\lim_{\varepsilon\rightarrow 0}\frac{p_{0,\mathrm{mrc}}^*}{p_{0,\mathrm{ssc}}^*}>1$ for $M$, thus we proceed from \eqref{limp} as follows
	\begin{align}
	(M!)^{-1/M}\frac{p_t\!+\!Mp_r\!+\!p_{syn}}{p_t\!+\!p_r\!+\!p_{syn}}&\! >\! 1 \nonumber\\
	p_t\!+\!Mp_r\!+\!p_{syn} & \!>\! (M!)^{1/M} (p_t\!+\!p_r\!+\!p_{syn}) . \label{ineL}
	\end{align}
	Notice that isolating $M$ in \eqref{ineL} is analytically intractable mainly because of the the tangled analytical dependence of the function $g(M)=(M!)^{1/M}$ on $M$. However, we find out that $g(M)$	 exhibits a nearly linear relation with $M$ given by
	\begin{align}
	g(M)\approx 0.374M+0.786.\label{gM}
	\end{align}
	The coefficients were obtained by using linear curve-fitting in the interval $1\le M\le 64$ and the accuracy of such approximation is shown in Fig.~\ref{Fig7}, where it is also observed that \eqref{gM} is still accurate for $M>64$. Finally, substituting \eqref{gM} into \eqref{ineL} and solving for $M$ we attain \eqref{theS}.
		\begin{figure}[t!]
		\centering
		\subfigure{\includegraphics[width=0.45\textwidth]{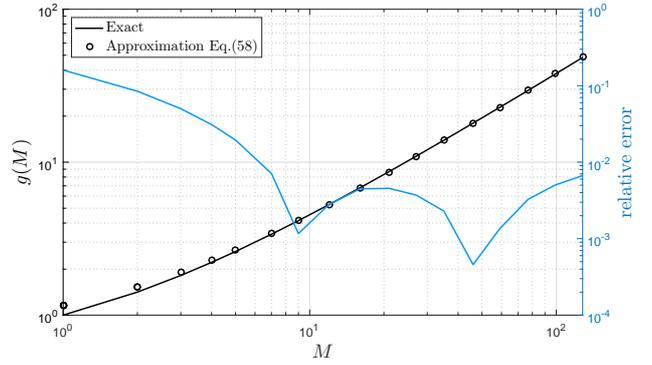}}
		\vspace*{-2mm}
		\caption{Exact and approximate curves for  $g(M)$. Relative error of the approximation.}	
		\vspace*{-2mm}
		\label{Fig7}
	\end{figure}
\end{proof}
Obviously if \eqref{theS} holds, the power consumption under the SSC scheme would be smaller than with MRC. If the overall  power consumption gap is greater than the gap in the optimum transmit rate given in \eqref{limr}, then SSC will be also more energy efficient. However, characterizing analytically such trade-off is cumbersome, hence we resort to numerical methods next.
\section{Numerical Analysis}\label{results}
Numerical results are presented in this section. Unless specified otherwise we set $\delta=10$ dB and on the basis of the power consumption
values found in \cite{Shuguang.2004}, we set $p_t=50$ mW ($-13$ dB), $p_r=60$ mW ($-12$ dB) and $\eta=0.35$; while $p_{syn}=10$ mW ($-20$ dB), which is a reasonable value according to \cite{Krishna.2007}.  Additionally, we consider $\kappa=8$, $p_{\min}=10$ mW ($-20$ dB), $p_{\max}=10$ W ($10$ dB) and $r_{\min}=0.01$ bps/Hz.

Fig.~\ref{Fig8} shows the optimization results as functions of the target outage constraint  for receiving devices with $M\in\{1,2,4,8\}$ antennas. The topology under study consists of $\kappa=8$ interfering nodes causing $2^{-i}\ \mu$W, $i=1,...,\kappa$, of average interference to $R_0$, while the path loss of the typical link is set to $\lambda_0=10^{-5}$, thus, $\delta\approx 10$ ($10$ dB). We compare our analytical results with a Monte Carlo approach using $10^7$ SIR realizations and a brute force technique for finding the solution of \textbf{P1} for each of the diversity schemes. Since we use a 2-dimensional search, the method is extremely time consuming and only sufficiently accurate for $\varepsilon\ge 10^{-5}$, thus, the simulation was carried out only in that region. 
We can notice that
\begin{itemize}
		\item Monte Carlo simulation results match accurately with our analytical results, hence validate our work;
	\item operating with only one antenna is practically unfeasible for the region where $\varepsilon<10^{-4}$ is required, while as the number of antennas increases we can operate with extremely high reliability with even relatively large data rates and reduced power consumption, thus, greater energy efficiency;
	\item the more stringent the reliability constraints,  more transmit power requires to be allocated while reducing the transmit rate and the optimum energy efficiency of the system. However notice that the curve slopes tend to 0 as the number of antennas increases. For instance, for $M=1$ yields 
	\begin{align}
	p_0^*\big|_{\varepsilon=10^{-1}}-p_0^*\big|_{\varepsilon=10^{-4}}&=-15\ \mathrm{dB}\nonumber\\
	r_0^*\big|_{\varepsilon=10^{-1}}/r_0^*\big|_{\varepsilon=10^{-4}}&=31\ (14.9\ \mathrm{dB}) \nonumber\\
	\mathrm{EE}^*\big|_{\varepsilon=10^{-1}}-\mathrm{EE}^*\big|_{\varepsilon=10^{-4}}&=29\ \mathrm{dB}\nonumber
	\end{align} 
	while already for $M=8$ and SC the variations are not very considerable since  
	\begin{align}
	p_0^*\big|_{\varepsilon=10^{-1}}-p_0^*\big|_{\varepsilon=10^{-8}}&=-4.5\ \mathrm{dB}\nonumber\\
	r_0^*\big|_{\varepsilon=10^{-1}}/r_0^*\big|_{\varepsilon=10^{-8}}&=2.5\ (4\ \mathrm{dB}) \nonumber\\
	\mathrm{EE}^*\big|_{\varepsilon=10^{-1}}-\mathrm{EE}^*\big|_{\varepsilon=10^{-8}}&=7.5\ \mathrm{dB}\nonumber
	\end{align}
\begin{figure}[t!]
	\centering
	\subfigure{\includegraphics[width=0.45\textwidth]{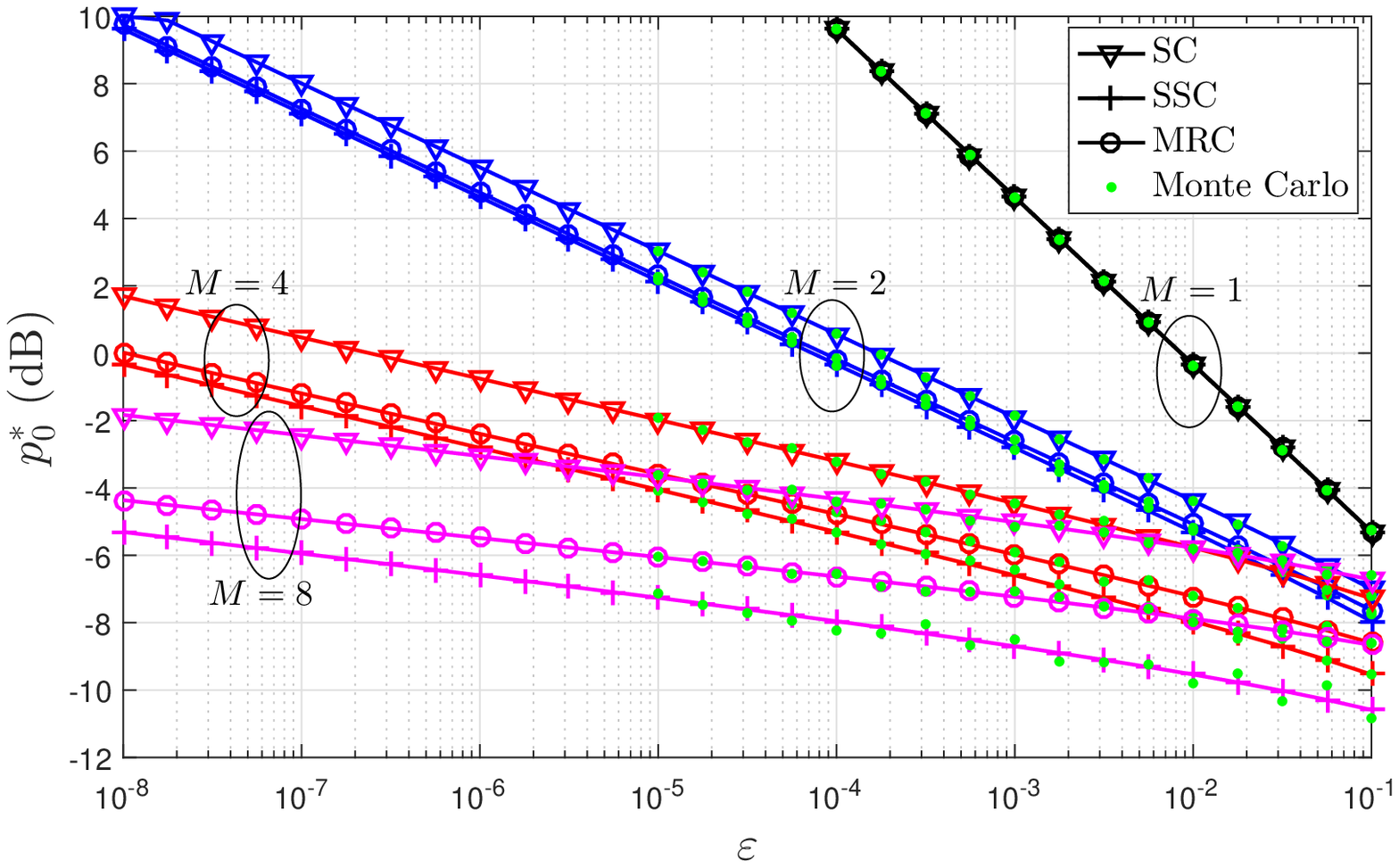}}\\ 
	\subfigure{\includegraphics[width=0.45\textwidth]{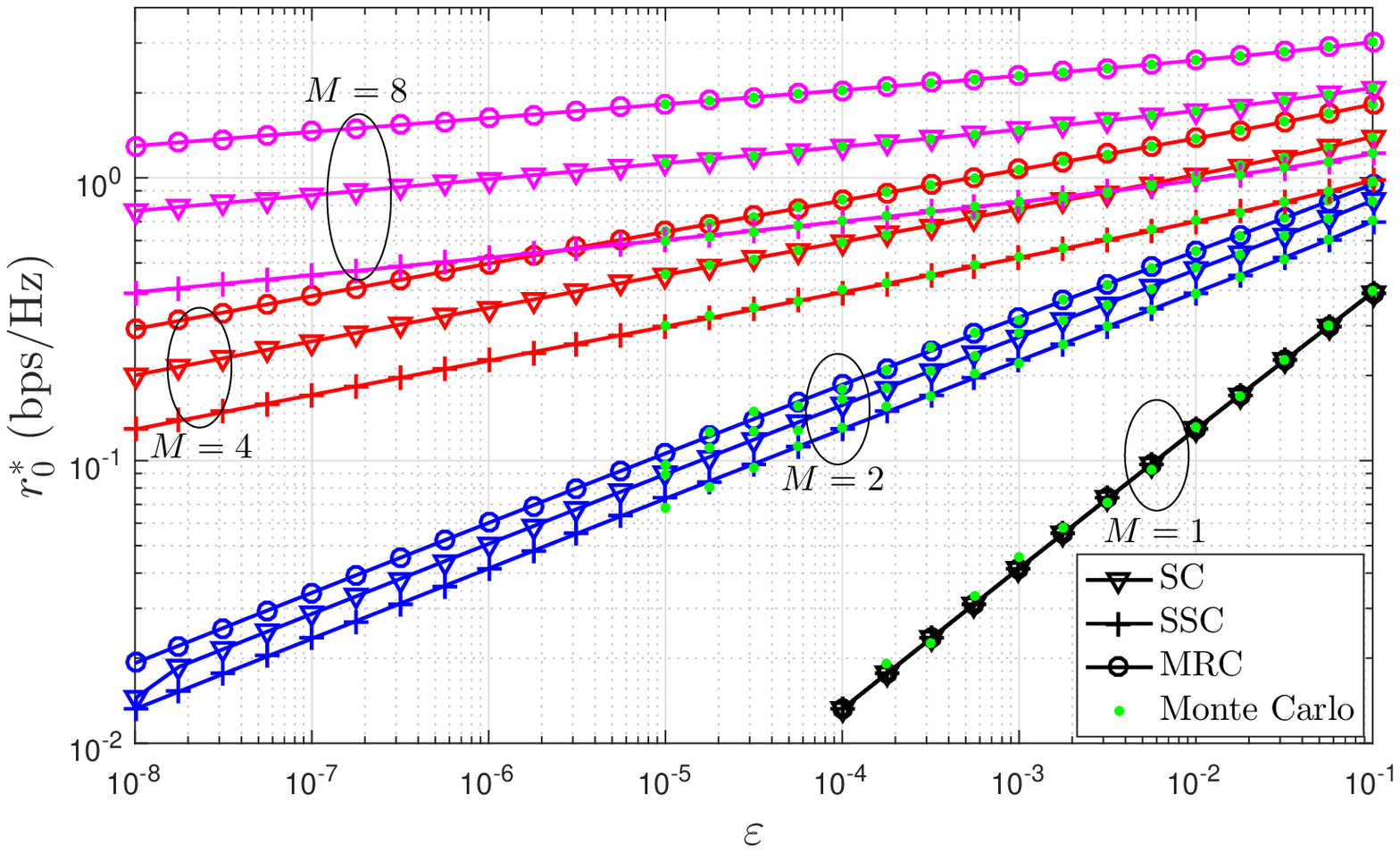}}\\
	\subfigure{\includegraphics[width=0.45\textwidth]{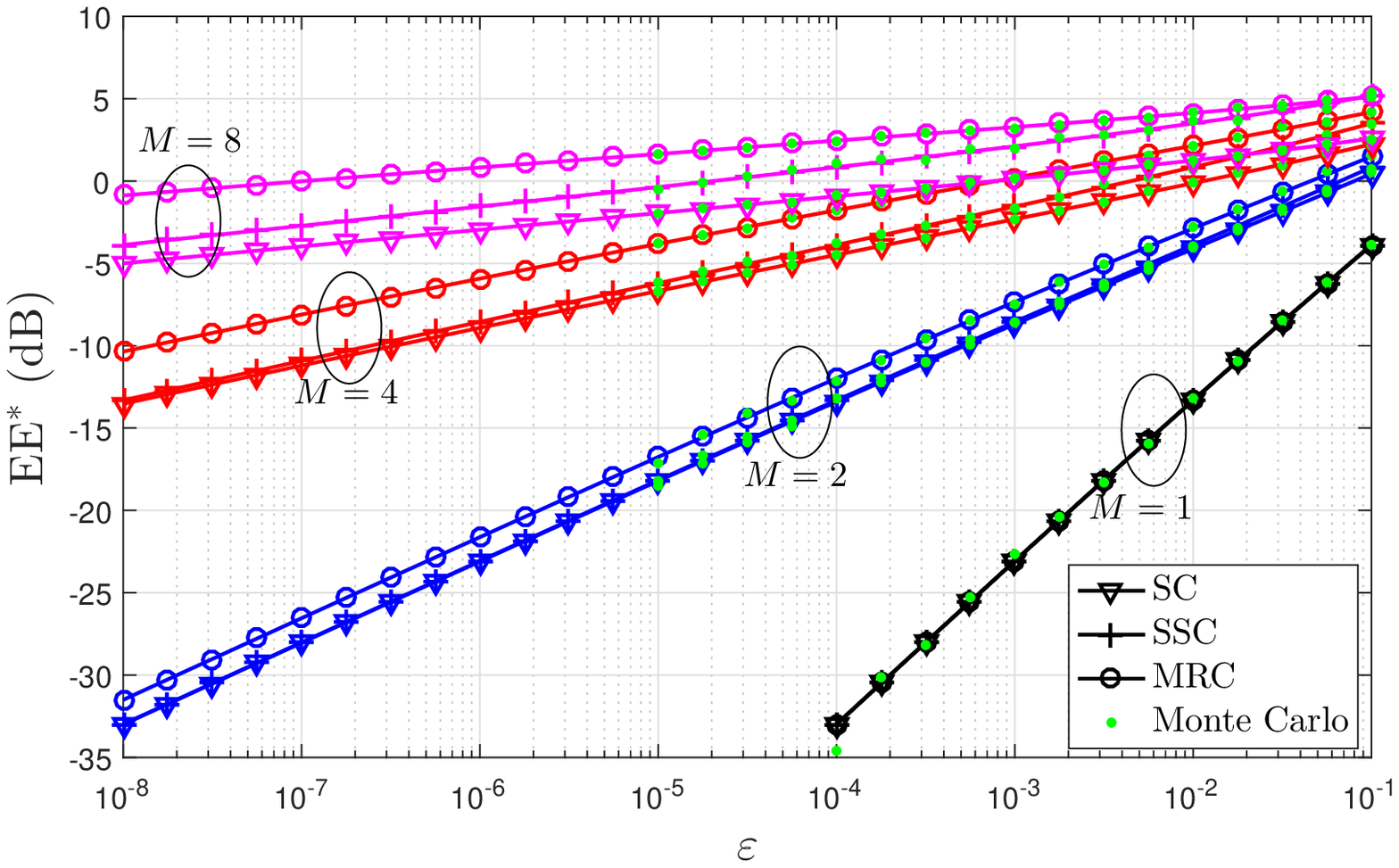}}
	\vspace*{-2mm}
	\caption{Optimization results $(a)$ $p_0^*$ (left-top), $(b)$ $r_0^*$ (right-top) and $(c)$ $\mathrm{EE}^*$ (bottom) as functions of $\varepsilon$ for $M\in\{1,2,4,8\}$ and SC, MRC diversity schemes. We set $\lambda_0=10^{-5}$ and $p_i\lambda_i=2^{-i}\ \mu$W for $i=1,\cdots,\kappa$.}	
	\vspace*{-2mm}	
	\label{Fig8}
\end{figure}
	\item as previously discussed in Subsection~\ref{SubC}, MRC is more energy efficient than the SC scheme since it requires less power while providing greater data rates to meet the same reliability constraint. Additionally, as the number of antennas increases, the gap in the performance metrics increases as predicted by the results for the asymptotic ultra-reliability regime. In fact, those results predict that
	\begin{itemize}
		\item for $M=4$, MRC allows operating with an optimum transmit rate/transmit power $1.4877$ (1.73 dB) times greater/smaller than what SC allows;
		\item for $M=8$, MRC allows operating with an optimum transmit rate/transmit power $1.94$ (2.88 dB) times greater/smaller than what SC allows;
	\end{itemize}
	which can be easily corroborated in Fig.~\ref{Fig8}a and Fig.~\ref{Fig8}b for $\varepsilon=10^{-8}$. Interestingly, that gap in the performance remains similar even for less stringent values of $\varepsilon$. 
	\item as previously discussed in Subsection~\ref{B1}, SSC is also more energy efficient than the SC scheme since it is able of achieving the same reliability performance using the same pair $(r_0,p_0)$ but with reduced power consumption. It turns out that the optimum transmit rate and power for SSC is smaller than for SC, and tends to $\sqrt{\frac{2}{M+1}}p_\mathrm{0,sc}^*$ as $\varepsilon\rightarrow 0$ according to \eqref{ga} and \eqref{ga2} for the chosen values of system parameters. As in the previous case, this gap characterization is analytically accurate as corroborated in Fig.~\ref{Fig8}a and Fig.~\ref{Fig8}b and even remains valid for not so stringent values of $\varepsilon$;
	\item results in Fig.~\ref{Fig8}c show MRC as more energy efficient than SSC, and the performance gap increases as the number of antennas increases and/or the reliability requirement becomes more stringent. According to  \eqref{theS}, as $\varepsilon\rightarrow 0$ SSC consumes less energy for transmitting than MRC for $M>2$, thus much smaller overall power consumption as $M$ increases as well. However, under MRC $T_0$ is able of transmitting with a larger data rate that overcompensates the loss in power consumption and consequently making the system more energy efficient.
\end{itemize}
\begin{figure}[t!]
	\centering
	\subfigure{\includegraphics[width=0.45\textwidth]{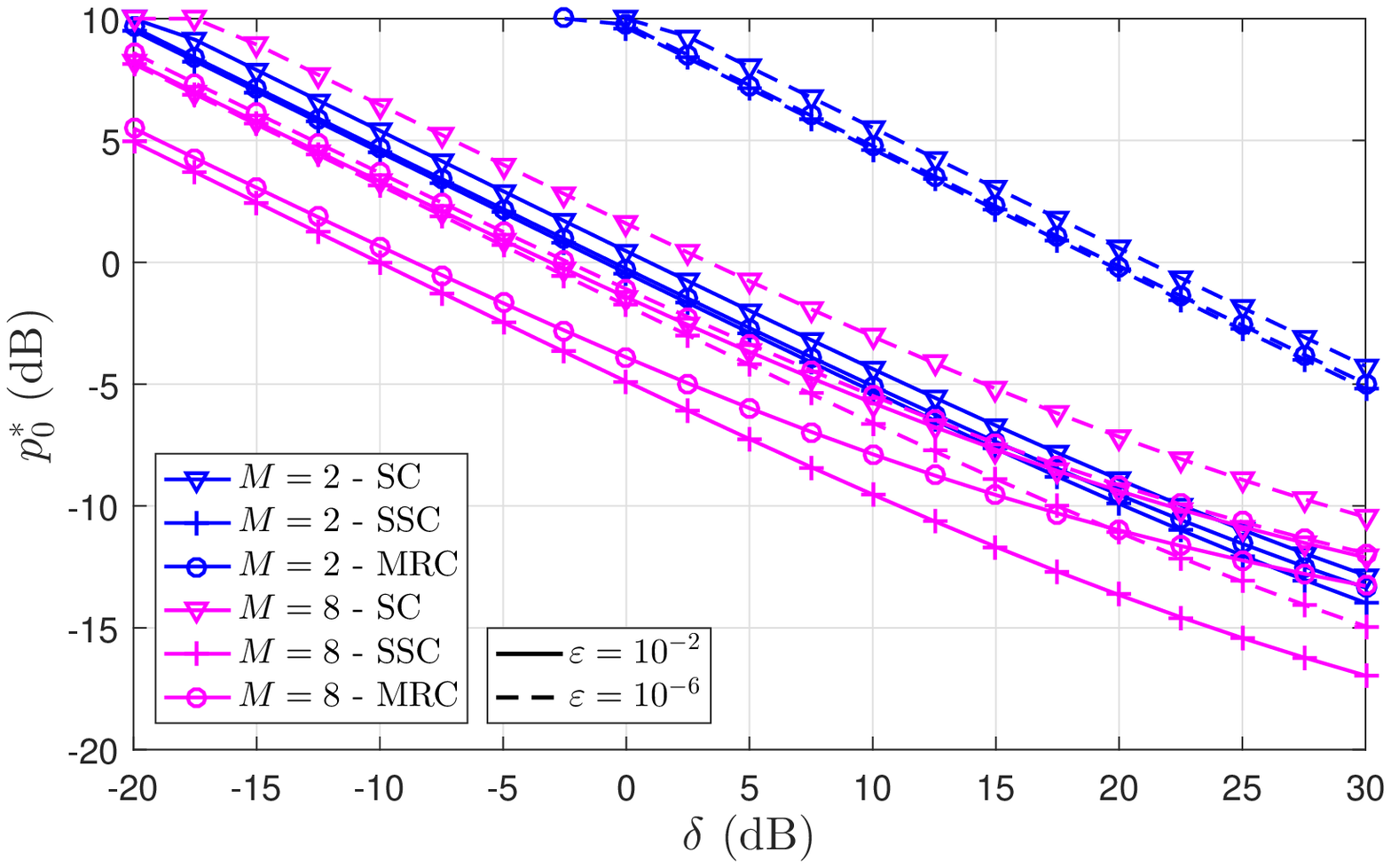}}\\
	\subfigure{\includegraphics[width=0.45\textwidth]{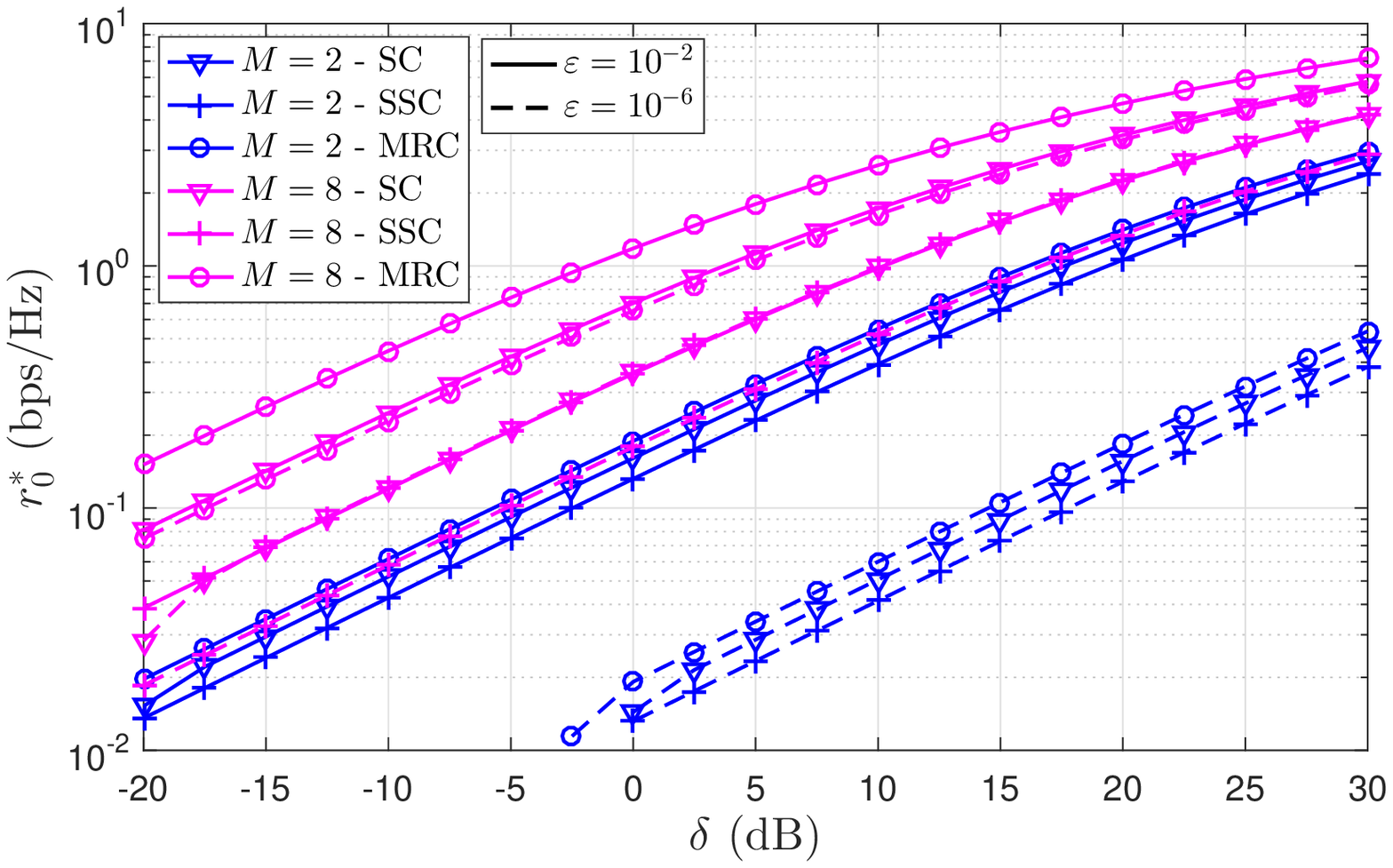}}\\
	\subfigure{\includegraphics[width=0.45\textwidth]{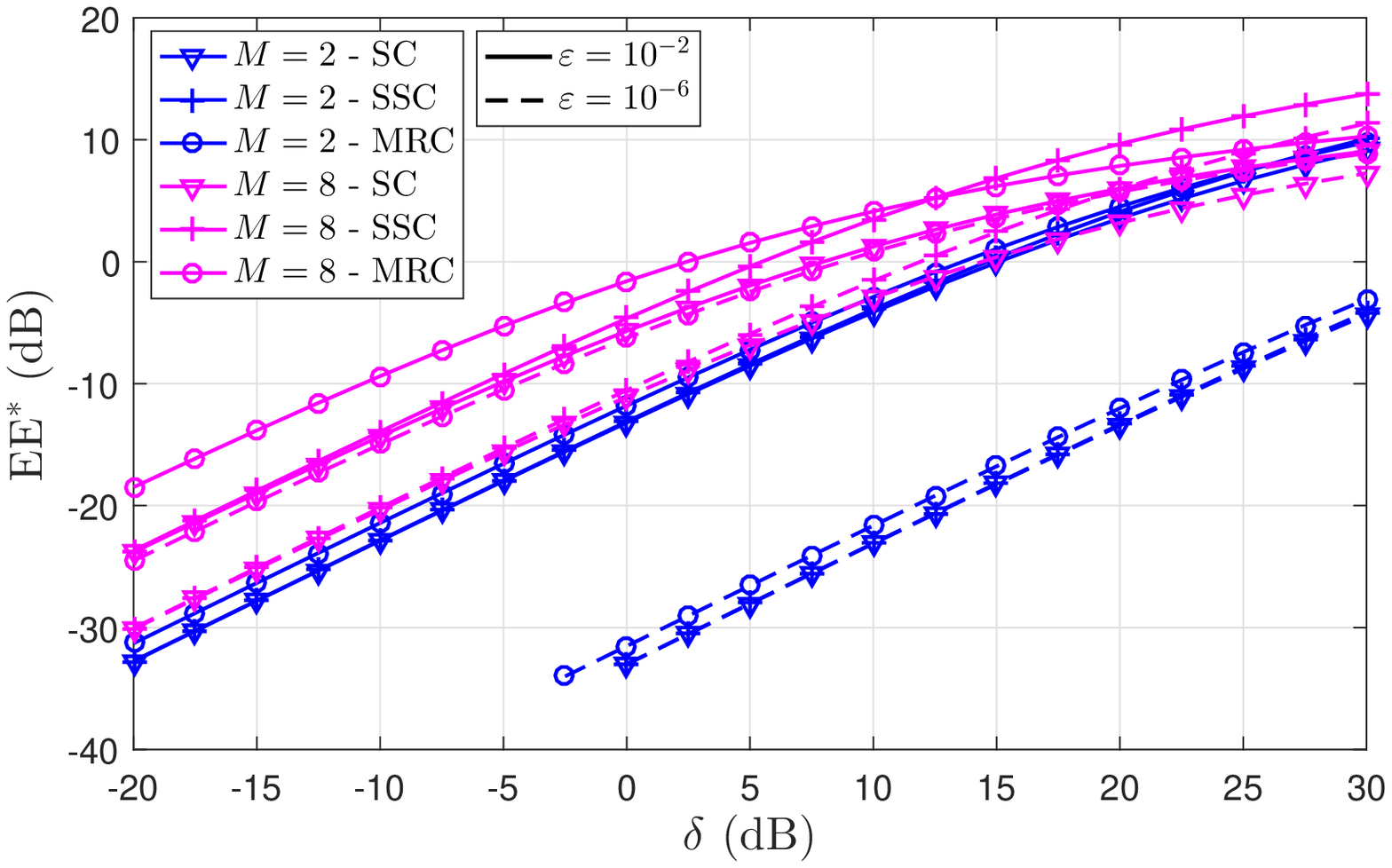}}
	\vspace*{-2mm}
	\caption{Optimization results $(a)$ $p_0^*$ (left-top), $(b)$ $r_0^*$ (right-top) and $(c)$ $\mathrm{EE}^*$ (bottom) as functions of $\delta$ for $M\in\{2,8\}$, $\varepsilon\in\{10^{-2},10^{-6}\}$ and SC, SSC, MRC diversity schemes.}	
	\vspace*{-2mm}	
	\label{Fig9}
\end{figure}

All the other remaining figures focus only on the results coming from evaluating the provided analytical expressions, therefore, they rely entirely on parameter $\delta$. In fact, Fig.~\ref{Fig9} shows the performance as a function of $\delta$. As $\delta$ increases, the optimum power decreases, while the optimum transmit rate and energy efficiency increases.  This is because an increment on $\delta$ is due to a smaller path loss in the typical link and/or smaller average perceived interference, thus, satisfying the reliability constraints more easily. From an analytical point of view, the greater $\delta$, the greater $\omega$ as shown in \eqref{omega}, thus, according to the discussion around Fig.~\ref{Fig5}, $p_0^*$ decreases, while $r_0^*$ and $\mathrm{EE}^*$ increase. Once again we can notice that the multi-antenna configuration enables the ultra-reliability operation even for very small values of $\delta$. Notice that the superiority of the MRC and SSC schemes over SC is evidenced again, while SSC is more energy efficient than MRC for large $\delta$ and even more as $M$ increases. This is because satisfying the reliability constraints becomes easier and the extra power consumption that would come from utilizing the entire set of antennas, as it is the case when using the SC or MRC schemes, does not bring great benefits.
\begin{figure}[t!]
	\centering
	\subfigure{\includegraphics[width=0.45\textwidth]{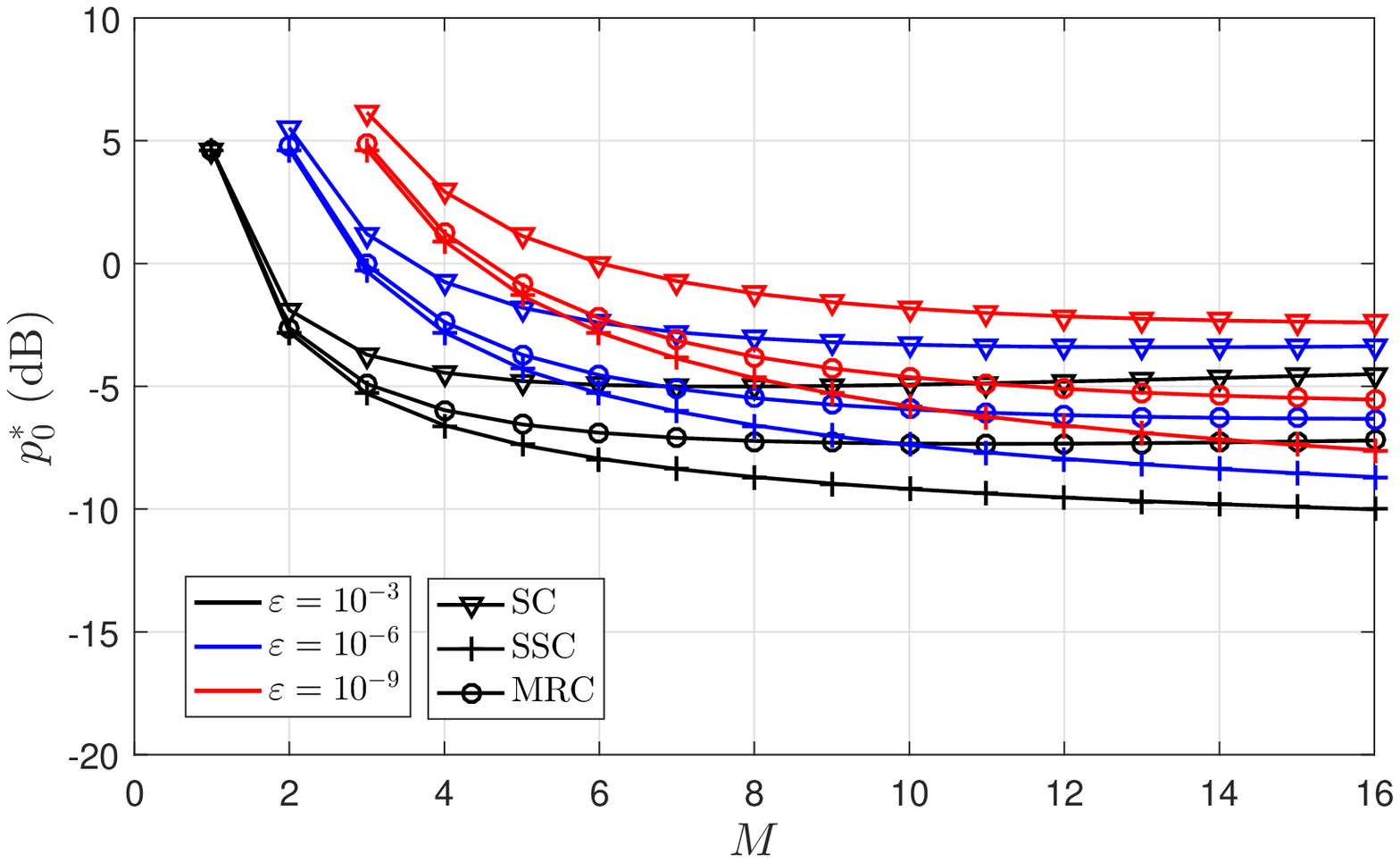}}\\
	\subfigure{\includegraphics[width=0.45\textwidth]{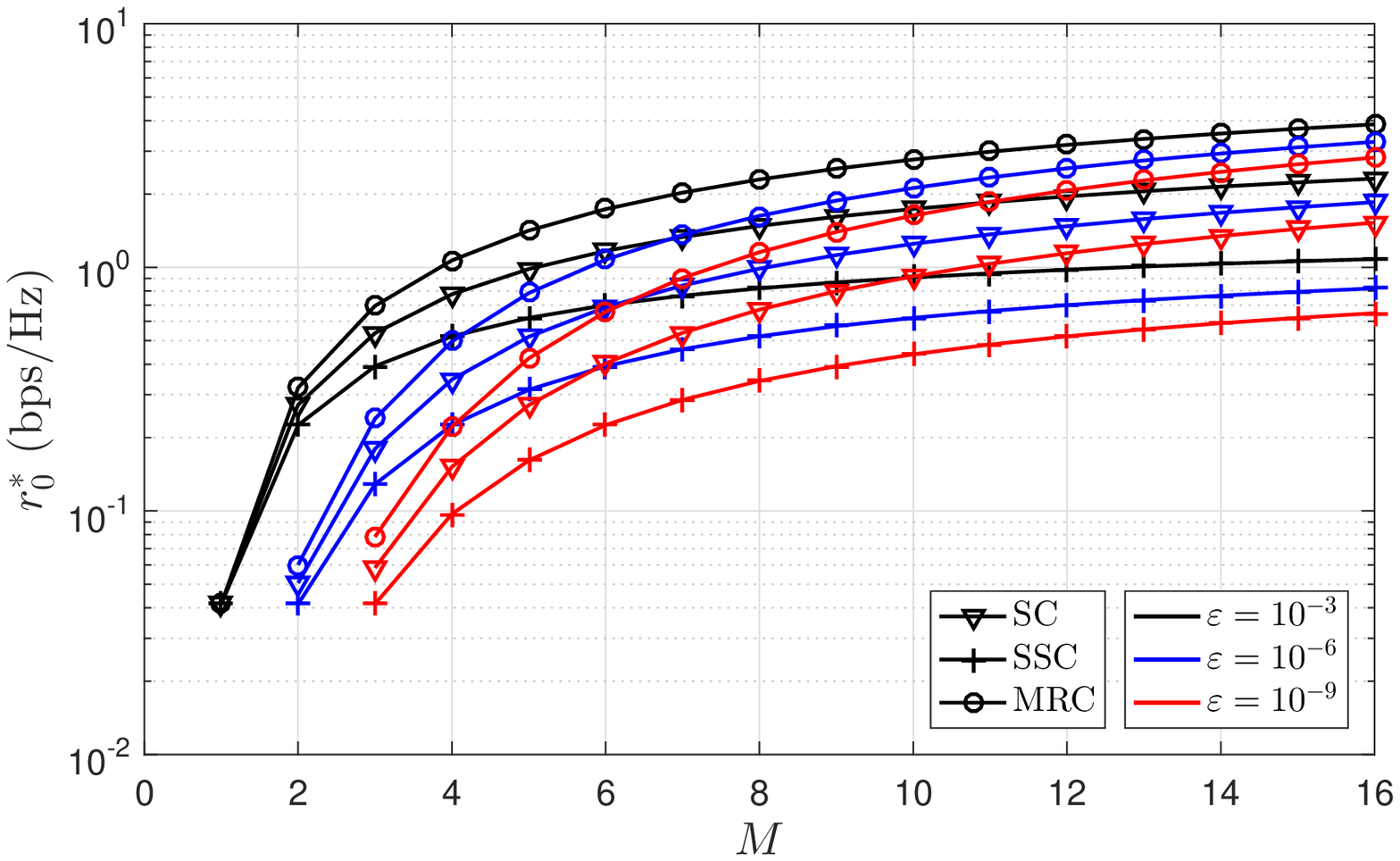}}\\
	\subfigure{\includegraphics[width=0.45\textwidth]{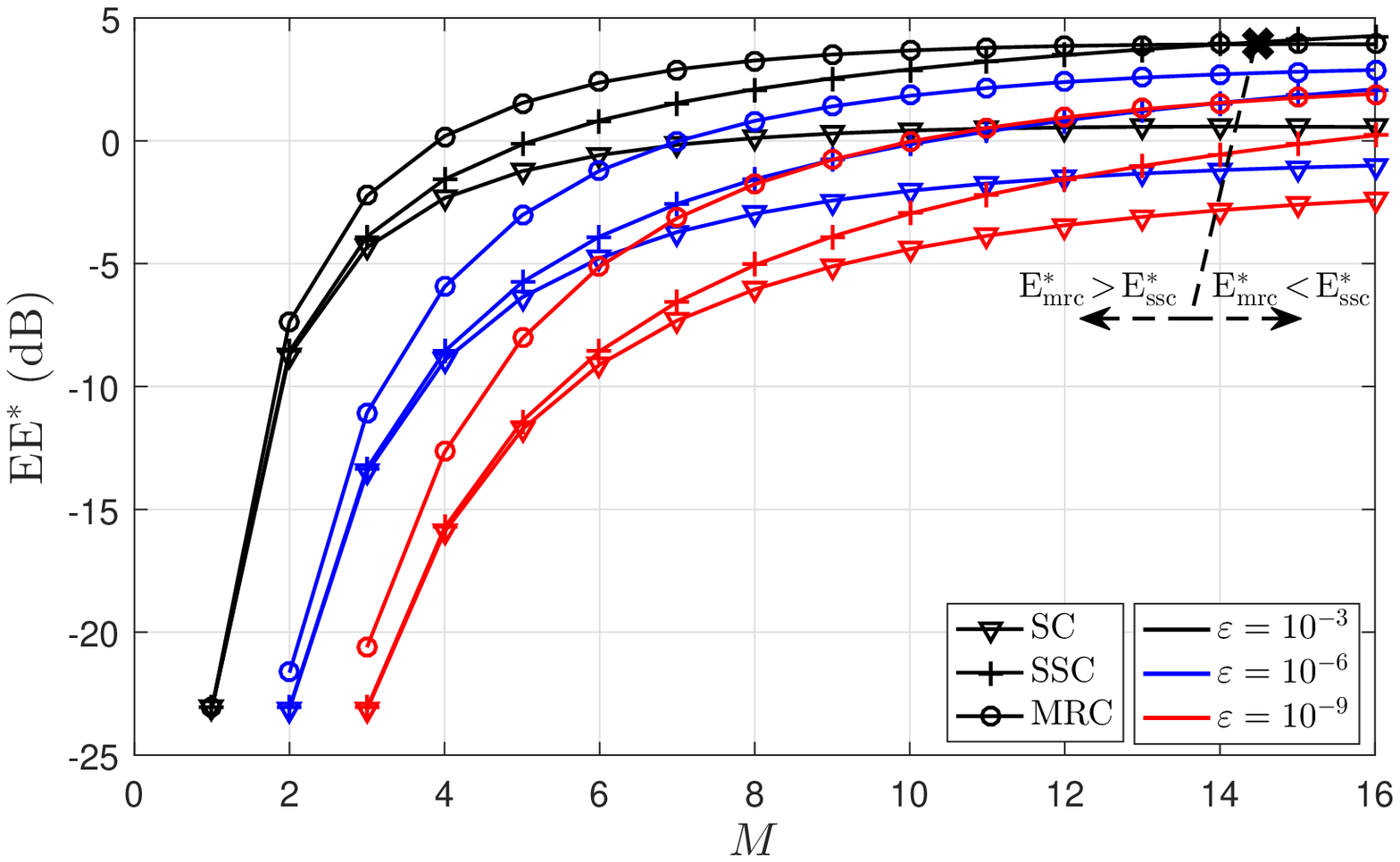}}
	\vspace*{-2mm}
	\caption{Optimization results $(a)$ $p_0^*$ (left-top), $(b)$ $r_0^*$ (right-top) and $(c)$ $\mathrm{EE}^*$ (bottom) as functions of $M$ for $\varepsilon\in\{10^{-3},10^{-6},10^{-9}\}$, and SC, SSC, MRC diversity schemes.}
	\vspace*{-2mm}
	\label{Fig10}
\end{figure}
Meanwhile, the impact of the multiple antennas at $R_0$ is shown with details in Fig.~\ref{Fig10} for $\varepsilon\in\{10^{-3},10^{-6},10^{-9}\}$. These results validate the discussion carried out for Fig.~\ref{Fig8} when analyzing the impact of $M$. This is i) the curve slopes tend to 0 as the number of antennas increases; ii) under SSC $T_0$ transmits with the smallest rate and power and the gap tends to widen as $M$ increases, e.g., according to $\sqrt{(M+1)/2}$ when comparing to SC, and $\sqrt{g(M)(M+1)/2}$, $\sqrt{(M+1)/(2g(M))}$ when comparing to MRC, respectively\footnote{Notice that these gaps are obtained for $\varepsilon\rightarrow 0$ and fixed $M$; however, as $M$ increases, $\varepsilon^{1/M}$ increases as well and according to \eqref{omega} $\omega$ decreases slower on $\varepsilon$; therefore, the larger $M$ the smaller $\varepsilon$, hence those gaps are accurate.}; iii) MRC is more energy efficient than SC and the performance gap between these schemes increases as $M$ increases. While in Fig.~\ref{Fig8} we showed that MRC overcame SSC as well, Fig.~\ref{Fig10} illustrates that the energy efficiency performance gap decreases with $M$ but also SSC could become even more energy efficient, specially when operating with not so stringent reliability requirements, e.g., $\mathrm{E_{ssc}}^*>\mathrm{E_{mrc}}^*$ for $M\ge 15$ and $\varepsilon=10^{-3}$.
\begin{figure}[t!]
	\centering
	\subfigure{\includegraphics[width=0.45\textwidth]{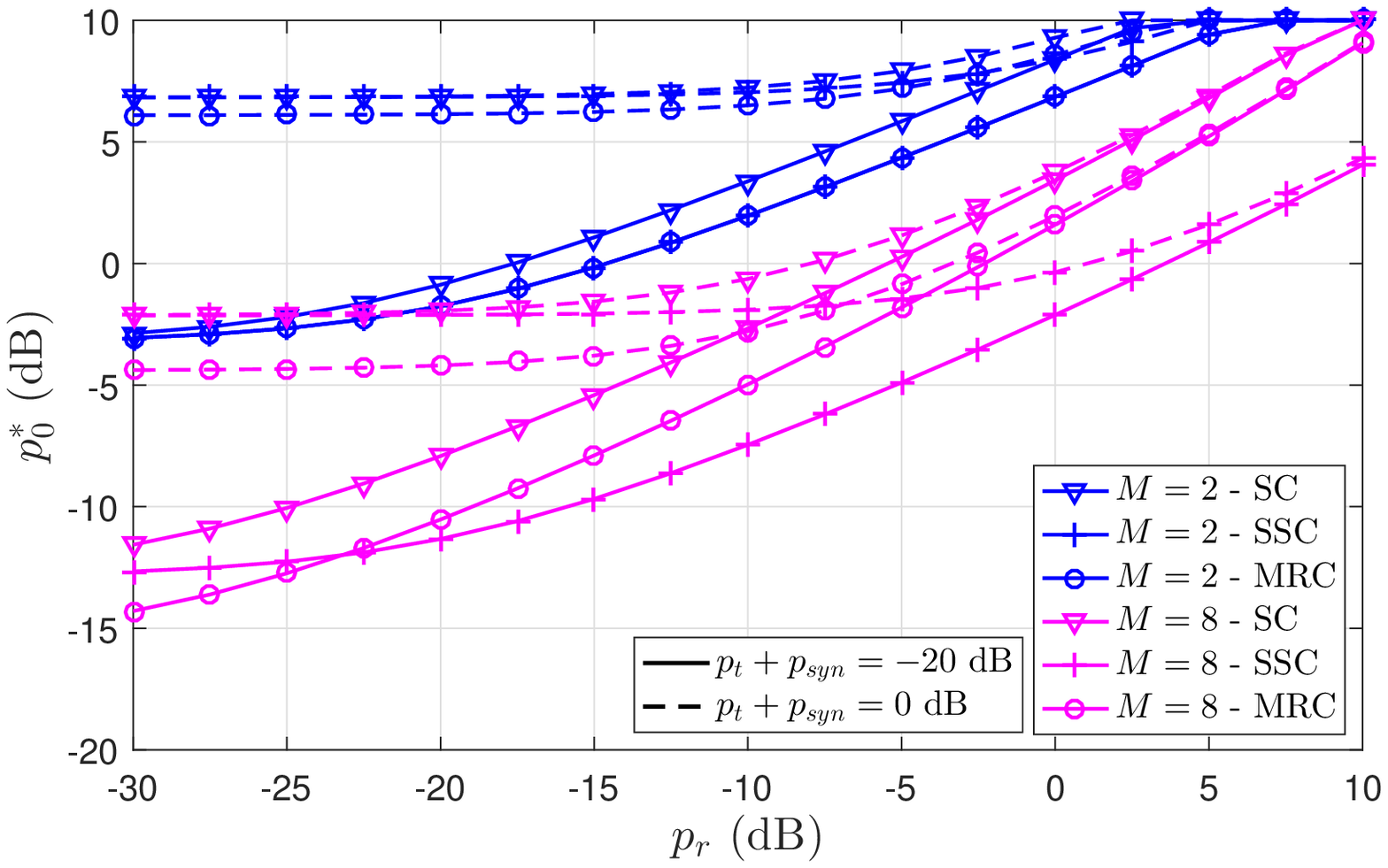}}\\
	\subfigure{\includegraphics[width=0.45\textwidth]{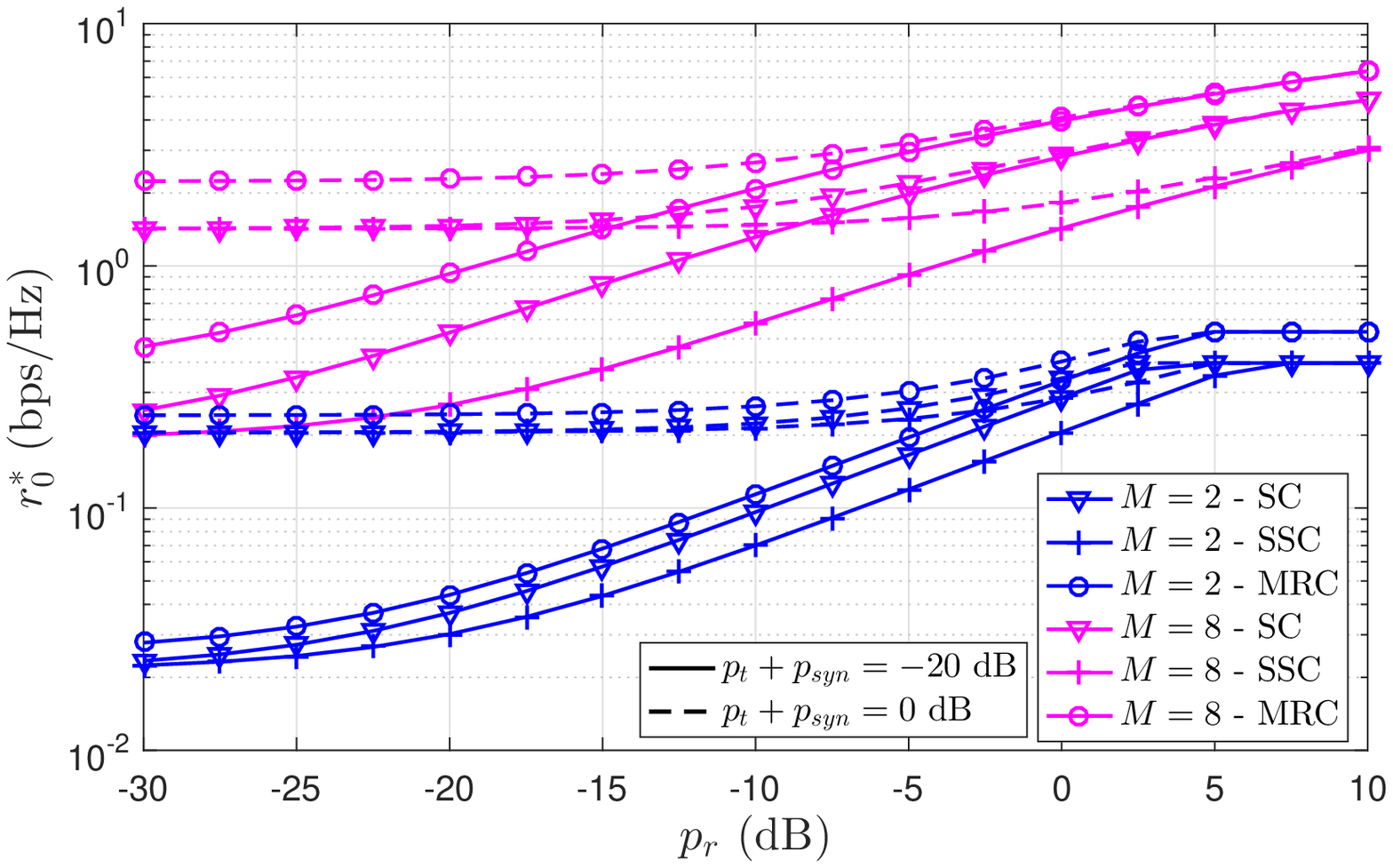}}\\
	\subfigure{\includegraphics[width=0.45\textwidth]{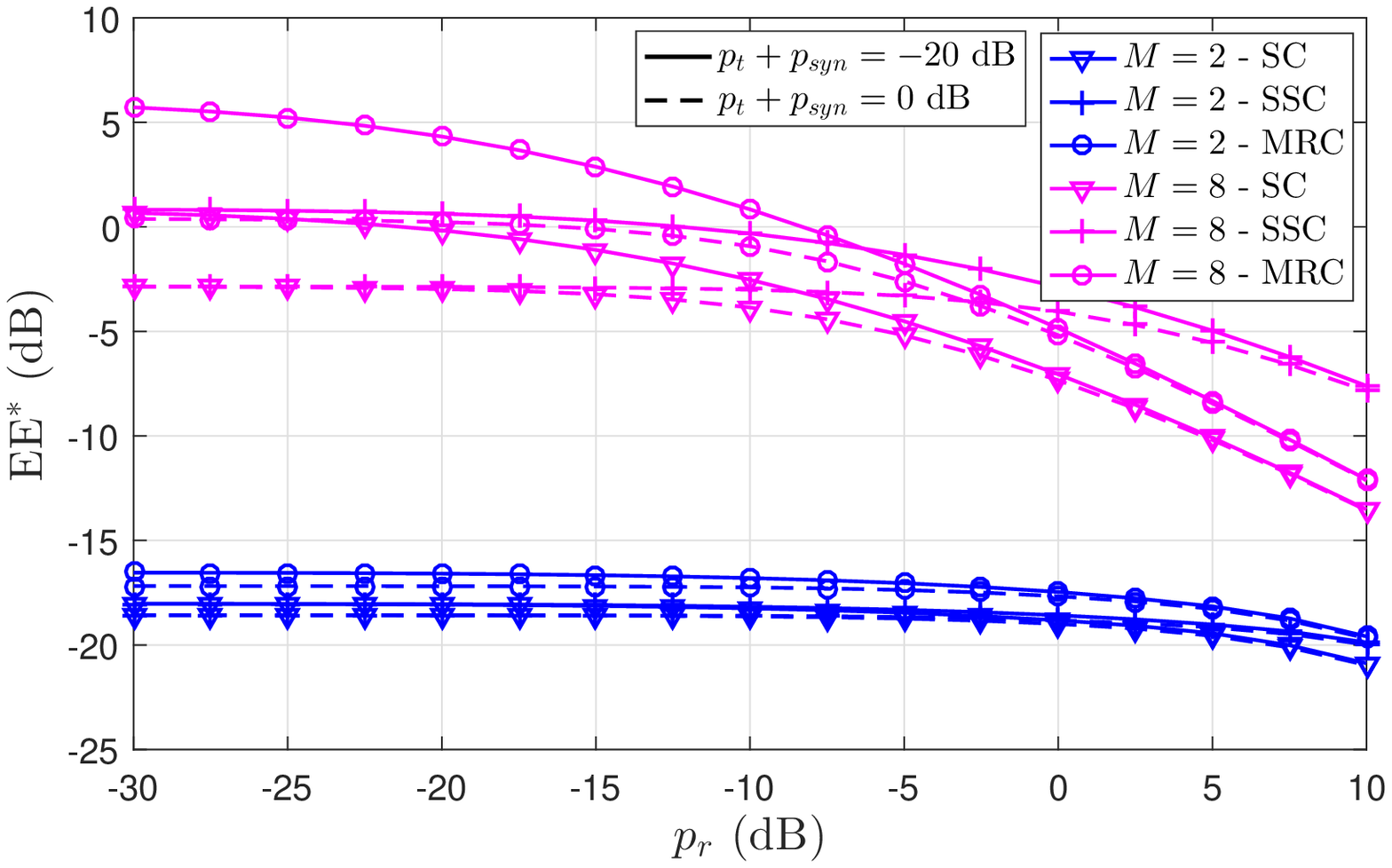}}
	\vspace*{-2mm}
	\caption{Optimization results $(a)$ $p_0^*$ (left-top), $(b)$ $r_0^*$ (right-top) and $(c)$ $\mathrm{EE}^*$ (bottom) as functions of $p_r$ for $M\in\{2,8\}$, $p_t+p_{syn}\in\{-20,0\}$ dB, $\varepsilon=10^{-5}$, and SC, SSC, MRC diversity schemes.}	
	\vspace*{-2mm}	
	\label{Fig11}
\end{figure}
\begin{figure}[t!]
	\centering
	\subfigure{\includegraphics[width=0.48\textwidth]{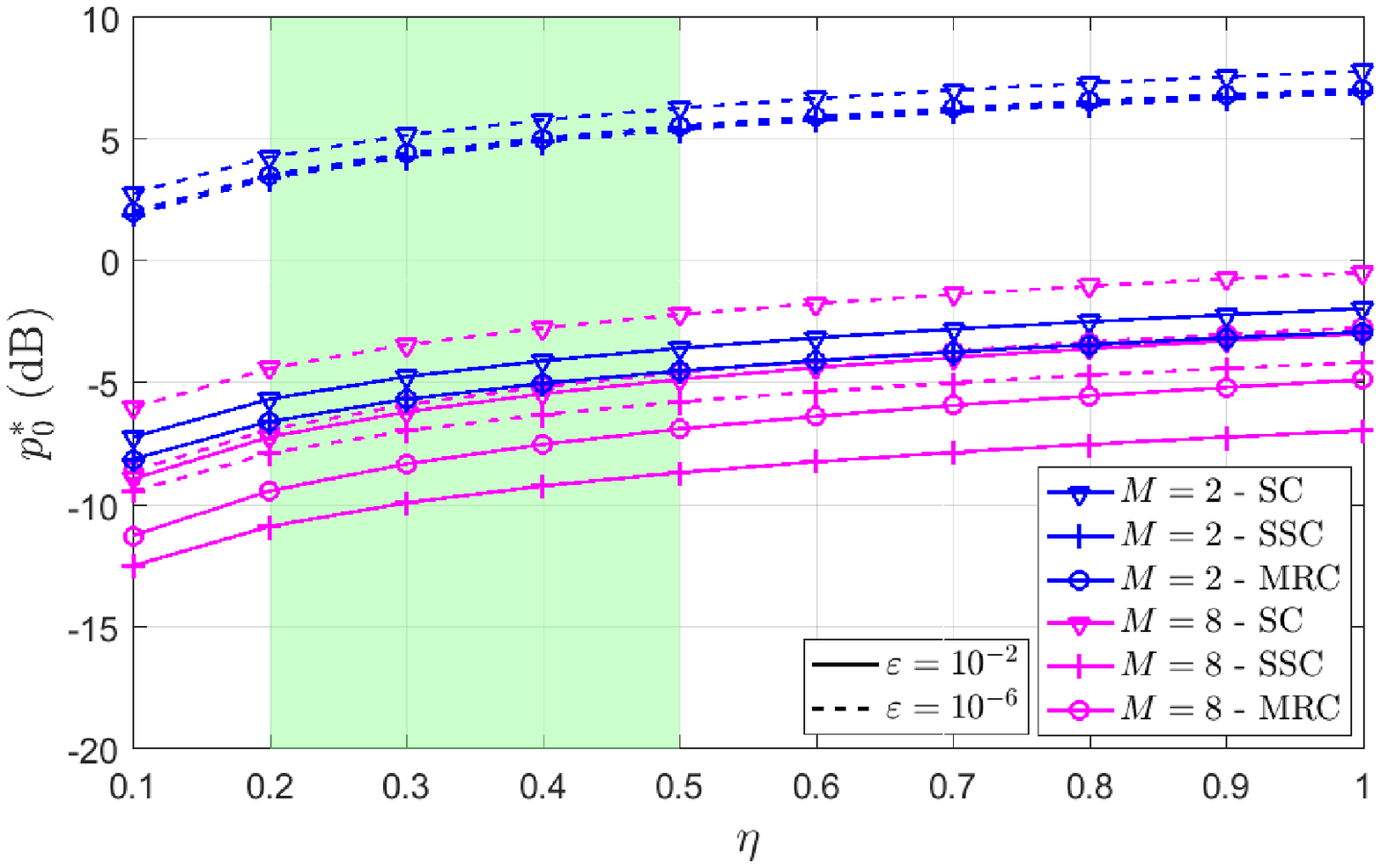}}\\
	\subfigure{\includegraphics[width=0.48\textwidth]{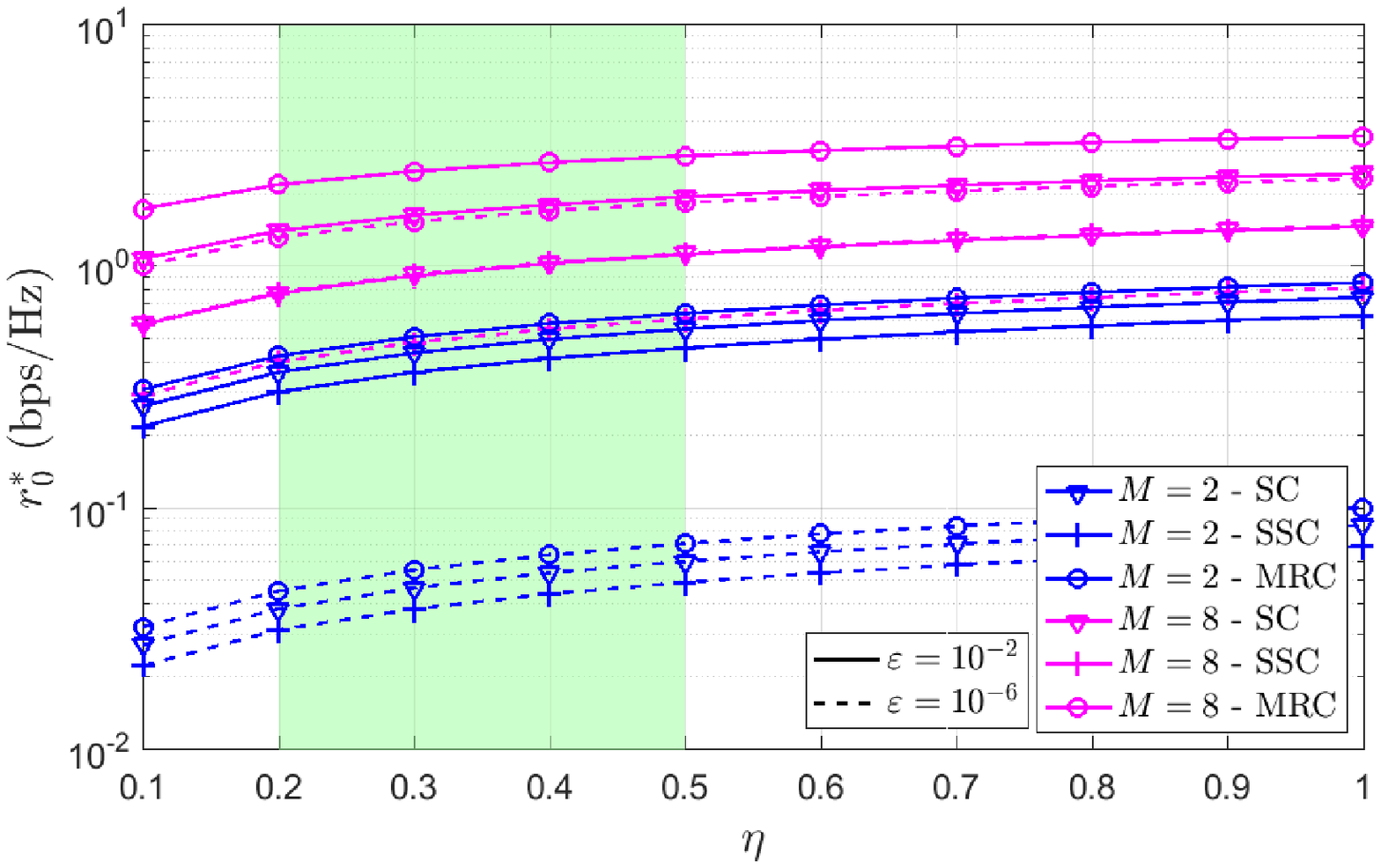}}\\
	\subfigure{\includegraphics[width=0.48\textwidth]{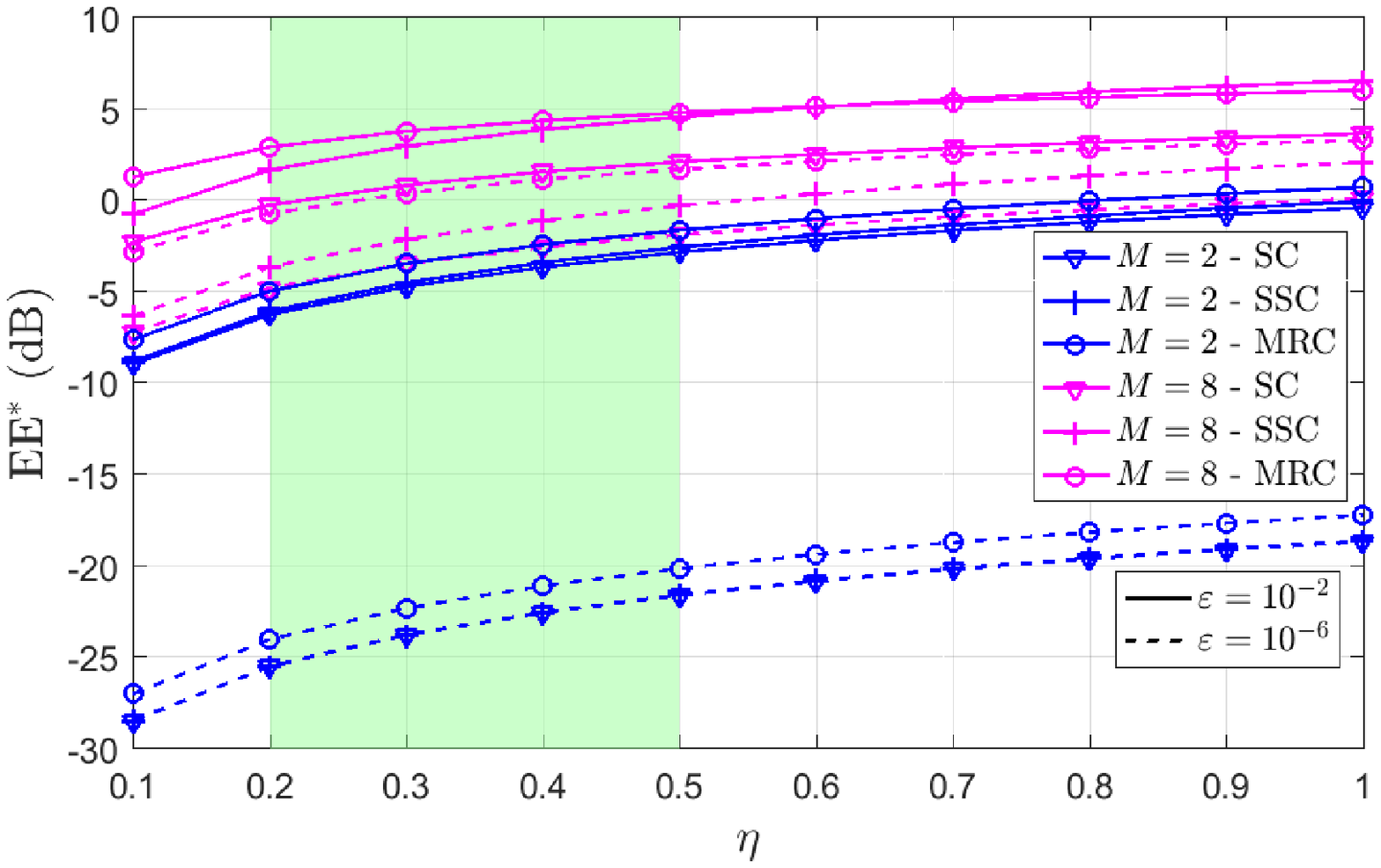}}
	\vspace*{-2mm}
	\caption{Optimization results $(a)$ $p_0^*$ (left-top), $(b)$ $r_0^*$ (right-top) and $(c)$ $\mathrm{EE}^*$ (bottom) as functions of $\eta$ for $M\in\{2,8\}$, $\varepsilon\in\{10^{-2},10^{-6}\}$ and SC, SSC, MRC diversity schemes.}
	\vspace*{-2mm}		
	\label{Fig12} 
\end{figure}

Fig.~\ref{Fig11} illustrates the performance as a function of the power consumption parameters. Since we now consider $p_t+p_{syn}$ and $p_r$ separately, the analysis here complements our previous discussions around Fig.~\ref{Fig5} where the overall circuitry power consumption was considered as a whole and intrinsically included the effect of $M$. As shown in Fig.~\ref{Fig11}, as $p_r$ increases the system energy efficiency is increasingly affected, specially when operating with large $M$ since for fixed $p_t+p_{syn}$ the circuitry consumption increases linearly with $p_r$ and $M$. As expected, as $p_r$ increases the SSC scheme becomes the most energy efficient. 
Finally, according to Fig.~\ref{Fig12}, the greater the drain efficiency of the amplifier at $T_0$, the greater the optimum transmit power, data rate and energy efficiency. This result is very interesting since so far an increment in the optimum transmit power conduced to a decrease in the optimum transmit rate and optimum energy efficiency, while variations in $\eta$ affect the three parameters, $p_0^*,\ r_0^*$ and $\mathrm{EE}^*$, in similar way. Notice that practical power amplifiers usually operate in the region $0.2 \lesssim\eta\lesssim 0.5$ \cite{Joung.2015}, and according to Fig.~\ref{Fig12}c the energy efficiency performance in those limits differs in around $2-4$ dB, which is substantial and it raises the need of efficient transmit hardware.
\section{Conclusion}\label{conclusions}
In this paper, we proposed a joint power control and rate allocation scheme that meets the stringent reliability constraints of the system while maximizing its energy efficiency.
The allocated resources depend on easy to obtain information such as i) $\delta p_0$, which is the ratio between the average signal and the average interference power at the receiver; ii) $\kappa$, the number of interfering transmitters; iii) $\varepsilon$, the reliability constraint; and iv) $M$, the number of antennas that are available at the receiver side as well as the diversity scheme, SC, SSC or MRC. 
We show the superiority of the MRC scheme with respect to SC in terms of energy efficiency 
since it allows operating with greater/smaller transmit rate/power. In that sense, we have proved that the gap between the optimum allocated resources for SC and MRC converges to $(M!)^{1/(2M)}$ under ultra-reliability constraints. Additionally, the optimum transmit rate and power are smaller when operating with SSC than with SC, and the ratio gap tends to be inversely proportional to the square root of a linear function of $M$; however, such allocation provides positive gains in the energy efficiency performance. Meanwhile, in most cases MRC was also shown to be more energy efficient than SSC, although this does not hold only when operating with extremely large $M$, $\delta$ and/or highly power consuming receiving circuitry. Numerical  results show the feasibility of the ultra-reliable operation when the number of antennas increases, while the greater the fixed power consumption and/or drain efficiency of the transmit amplifier, the greater the optimum transmit power and rate.

\bibliographystyle{IEEEtran}
\bibliography{IEEEabrv,references}
\end{document}